\documentclass[journal, onecolumn]{IEEEtran}
\usepackage{cite}
\usepackage{amsthm}
\usepackage{amsmath}
\usepackage{comment}
\usepackage{booktabs}
\usepackage{stfloats}
\usepackage{bm}
\usepackage{setspace}
\usepackage{graphicx}
\usepackage[usenames,dvipsnames]{color}
\usepackage{amsfonts,amssymb,amsmath}

\usepackage{enumerate}
\usepackage[shortlabels]{enumitem}

\usepackage{caption}

\usepackage{siunitx}

\usepackage{cases}
\usepackage{array}
\usepackage{colortbl}
\usepackage{multirow}
\usepackage{arydshln}

\theoremstyle{remark}

\usepackage{graphicx}
\usepackage{float}
\usepackage{subfig}

\newtheorem{lemma}{\textbf{{Lemma}}} 
\newtheorem{Theorem}{\textbf{{Theorem}}}

\definecolor{Gray}{gray}{0.85}
\definecolor{mycyan}{cmyk}{.3,0,0,0}

\newcolumntype{a}{>{\columncolor{Gray}}c}%\centering\backslash
\newcolumntype{b}{>{\columncolor{white}}c}
\newcolumntype{d}{>{\columncolor{mycyan}}c}

\allowdisplaybreaks

\begin{document}
	\captionsetup[figure]{labelfont={},name={Fig.},labelsep=period} 
	
	\title{Analysis of Age of Information for A  Discrete-time Dual-Queue System}
	\author{Zhengchuan Chen,~\IEEEmembership{Senior Member,~IEEE}, Yi~Qu, Nikolaos Pappas,~\IEEEmembership{Senior Member,~IEEE}, Chaowei~ Tang,~\IEEEmembership{Member,~IEEE}, Min~ Wang,~\IEEEmembership{Member,~IEEE}, and Tony Q. S. Quek,~\IEEEmembership{Fellow,~IEEE} 
		
		%\thanks{This paper was presented in part at the IEEE VTC2023-Fall \cite{Yang2023age}.
			%This paper is supported in part by  the National Natural Science Foundation of China under Grant 62271092 and Grant 62201504, in part by the Natural Science Foundation of Chongqing, China under Grant NO. CSTB2023NSCQ-MSX0933 and Grant NO. CSTB2022NSCQ-MSX0327, in part by the Fundamental Research Funds for the Central Universities under Project No. 2023CDJKYJH044, in part by the Xiaomi Young Talents Program, in part by the Zhejiang Provincial Natural Science Foundation of China under Grant LGJ22F010001, in part by the open research fund of National Mobile Communications Research Laboratory, Southeast University (No. 2024D05), in part by the Zhejiang – Singapore Innovation and AI Joint Research Lab, and in part by the National Research Foundation, Singapore and Infocomm Media Development Authority under its Future Communications Research \& Development Programme. %\emph{(Corresponding Author: )}	}
		\thanks{Z. Chen, Y. Qu, and C. Tang are with the School of  Microelectronics and Communication Engineering, Chongqing University, Chongqing, China (E-mails:~\{czc@cqu,~yqu@stu.cqu,~cwtang@cqu\}.edu.cn).}
		\thanks{N. Pappas is with the Department of Computer and Information Science, Linköping University, 58183 Linköping, Sweden (Email: nikolaos.pappas@liu.se).}
		%\thanks{H. H. Yang is with the Zhejiang University/University of Illinois UrbanaChampaign Institute, Zhejiang University, Haining 314400, China, and also with the National Mobile Communications Research Laboratory, Southeast University, Nanjing 211111, China (email: haoyang@intl.zju.edu.cn).}
		\thanks{M. Wang is with the School of Optoelectronics Engineering, Chongqing University of Posts and Telecommunications, Chongqing, China (E-mail:~wangm@cqupt.edu.cn).}
		\thanks{T. Q. S. Quek is with the Singapore University of Technology and Design, Singapore 487372, and also with the Yonsei Frontier Lab, Yonsei University, South Korea (e-mail: tonyquek@sutd.edu.sg).}
	}
	\maketitle

	\begin{abstract}
		Using multiple sensors to update the status process of interest is promising in improving the information freshness. The unordered arrival of status updates at the monitor end poses a significant challenge in analyzing the timeliness performance of parallel updating systems. This work investigates the age of information (AoI) of a discrete-time dual-sensor status updating system. Specifically, the status update is generated following the zero-waiting policy. The two sensors are modeled as a geometrically distributed service time queue and a deterministic service time queue in parallel. We derive the analytical expressions for the average AoI and peak AoI using the graphical analysis method. Moreover, the connection of average AoI between discrete-time and continuous-time systems is also explored. It is shown that the AoI result of the continuous-time system is a limit case of that of the corresponding discrete-time system. Hence, the AoI result of the discrete-time system is more general than the continuous one. Numerical results validate the effectiveness of our analysis and further show that randomness of service time contributes more AoI reduction than determinacy of service time in dual-queue systems in most cases, which is different from what is known about the single-queue system.   
	\end{abstract}
	
	
	\section{Introduction}
	\setlength{\textfloatsep}{0.9\baselineskip plus 0.2\baselineskip minus 0.2\baselineskip}
		Internet of Things (IoT) plays an increasing role in emerging industry such as e-health, smart home, auto-driving, and remote monitoring\cite{Iot}. For real-time applications in IoT systems, sensors  need to sense their physical surroundings and monitor system status in a timely manner to provide effective information for intelligent decision making and control. In such scenarios, the freshness of information is extremely important. The concept of age of information (AoI), which characterized the time elapsed since the generation of the latest received status update, becomes a pertinent  measure of information freshness. Since the emergency of AoI, abundant insights for improving information freshness of status updating system have been reported\cite{Aoigainian}.

%	Using multiple sensors to update the status process of interest is promising in improving the information freshness. The unordered arrival of status updates at the monitor end poses a significant challenge in analyzing the timeliness performance of parallel updating systems. This work investigates the age of information (AoI) of a discrete-time dual-sensor status updating system. Specifically, the status update is generated following the zero-waiting policy. The two sensors are modeled as a geometrically distributed service time queue and a deterministic service time queue in parallel. We derive the analytical expressions for the average AoI and peak AoI using the graphical analysis method.
%Moreover, the connection of average AoI between discrete-time and continuous-time systems is also explored. It is shown that the AoI result of the continuous-time system is a limit case of that of the corresponding discrete-time system. Hence, the AoI result of the discrete-time system is more general than the continuous one. Numerical results validate the effectiveness of our analysis and further show that randomness of service time contributes more AoI reduction than determinacy of service time in dual-queue systems in most cases, which is different from what is known about the single queue system. 

	Initial research on AoI mainly focused on the analysis of average and peak AoI of status updating in elementary classical queuing systems under different queuing disciplines\cite{firstwork, stationary,distribution1,distribution2,distribution3, LCFS,bao1,bao2,bao3,firstMulti,duoduilie1,duoduilie2,TCP,shuangduilie,mmmd}. In these studies, the transmission link between a sensor and a monitor is often abstracted as a queue. Yates et al. in \cite{firstwork} derived the average AoI of the M/M/1 \footnote{Kendall's notation for the classification of queue types. In particular, an asterisk superscript of the server number represents preemption.}, M/D/1 and D/M/1 queuing systems under the first-come-first-served (FCFS) discipline, to characterize the properties of the AoI starting from basic models. 
	% This work provides a basis for later description of state AoI in more complex real systems. 
	The authors of \cite{stationary} studied the average AoI of the M/GI/1 and GI/M/1 queues and derived a general expression of the stationary distribution, which provides a more comprehensive understanding of the AoI of systems. 
	%	Extensions of these basic queueing models are considered in \cite{distribution1,distribution2,distribution3}, considering  general formulas for the distribution of AoI and PAoI in a class of networks.
	The AoI under extended models of the basic queuing models is considered in \cite{distribution1,distribution2,distribution3}, and general formulas for the distribution of AoI and PAoI in a class of networks are obtained.
	For enhancing the performance of AoI for queuing systems, further works investigated different service disciplines and evaluated the corresponding performance gain\cite{LCFS,bao1,bao2,bao3}. Specifically, the authors of  \cite{LCFS} investigated the AoI of M/M/1/1 under the last-come-first-served (LCFS) discipline, and the results indicated that, compared with FCFS discipline, prioritizing newly arrived updates under the LCFS discipline can effectively reduce the average AoI of systems. Packet management is raised as an effective queuing policy for information freshness enhancement \cite{bao1,bao2,bao3}. Specifically, the work in  \cite{bao1} presented the average AoI of the M/M/1 queue under different packet management strategies in which the M/M/1 queue is equivalently modeled as M/M/1/1, M/M/1/2, and M/M/1/2* queues, respectively. The authors of \cite{bao2} further developed packet management strategies by introducing a penalty function, which %that incorporates the non-preemption with no buffer strategy, the non-preemption with one buffer strategy, and the preemption strategy
 help better understand the trade-off between system reliability and tolerance for outdated data. 
	The work in \cite{bao3}  proposed a mechanism to add service waiting in packet queues, which were modeled as M/GI/1/1 and M/GI/1/2* queues, revealing that discarding stale packets and replacing data packets with new status updates can further reduce the AoI. In order to further improve the information freshness, an intuitive idea is to use additional communication devices, e.g., sensors,  for status updating\cite{firstMulti,duoduilie1,duoduilie2,TCP,shuangduilie,mmmd}. Multiple servers are employed to process information in a single-queue system to enhance the information freshness\cite{firstMulti,duoduilie1,duoduilie2}. The average AoI of a multi-server system was first studied in \cite{firstMulti}. The work in \cite{duoduilie1} modeled systems with limited and abundant servers as M/M/2 and M/M/$\infty$ model, respectively, providing insights on how multi servers affect the AoI evolution of the system.  The work in \cite{duoduilie2} extended the study of AoI to more general M/M/c queue models. The results showed that the increase in the number of servers is quite significant in reducing the average AoI of the system. However, when the number of servers is increased to a certain number, adding more servers on top of that results in a very limited reduction in AoI. This phenomenon has important implications for balancing performance improvement and resource consumption. The effective reduction of the AoI of systems by a single queue with multiple servers promoted to expand the study to multiple queues with multiple servers\cite{TCP,shuangduilie,mmmd}. Motivated by the multi-path transmission technology,  the authors of \cite{TCP}  analyzed the average AoI of a status-updating network with multiple transmission queues. To further explore the impact of multi-queue transmission on AoI,	the works \cite{shuangduilie} and \cite{mmmd}  studied the synchronous and asynchronous dual-queue status updating system, respectively. It is shown that as the dual-queue transmission can send more packets containing status updates to the destination node, further reduction of the AoI can be obtained compared with just increasing the number of servers in single queue systems\cite{shuangduilie,mmmd}.
	
	While most existing works focused on continuous-time systems, a few recent studies paid close attention to the AoI analysis in discrete-time systems, providing more appropriate queuing modeling for widely adopted synchronized communication systems and time-slotted wireless networks. Some works studied the engineering application of designing discrete-time systems with AoI as the performance indicator\cite{wuren1,wuren2,aloha_multi,aloha_col,aloha_chan,aloha_pa}. The authors of \cite{wuren1,wuren2} considered slot-by-slot-operated unmanned aerial vehicle (UAV) systems and proposed different strategies to plan the paths of UAVs for AoI optimization, respectively. The works in \cite{aloha_multi,aloha_col,aloha_chan,aloha_pa} studied the update strategy and operating protocol in slotted time  ALOHA networks to improve network-level information freshness. The authors of \cite{aloha_multi} investigated the average AoI of time-slotted ALOHA in multi-access networks and adjusted the frame size and protocol-related parameters to improve the overall information freshness. Strategies for data sending in slotted ALOHA  with multi-user conflicts were investigated, and an AoI-oriented decision protocol was proposed to achieve overall multi-user freshness in \cite{aloha_col}. The work in  \cite{aloha_chan}	studied a decentralized transmission strategy for multiple transmitters in a random channel under time-slotted ALOHA with AoI as a metric.	On the other hand, the authors of \cite{aloha_pa} used PAoI as a metric to study the conflict problem caused by multi-node status updates. Meanwhile, several works 
	paid attention to the study of AoI of abstracted discrete-time queue systems\cite{lisanfcfs,lisanber,lisanlcfs,discretestationary,discreteMultiple}. In particular, the average AoI and PAoI of the Ber/Geo/1/1 queue status updates have been derived in closed-form \cite{lisanfcfs}. The authors of \cite{lisanber} further analyzed the distribution of PAoI and AoI in the Ber/Geo/1/1 queue and compared it with the M/M/1/1 queue to explore the correlation between continuous and discrete-time queues.
	%	 since  both geometric and exponential distributions have memory-less properties.
	The authors of \cite{lisanlcfs} applied the LCFS queuing discipline and packet management policies to discrete-time Ber/Geo/1 queues and derived the expressions for the corresponding average AoI and PAoI. Furthermore,  closed-form expressions for the stationary distribution of AoI and PAoI in discrete-time systems have been obtained for FCFS, LCFS, and packet management queuing disciplines\cite{discretestationary}. The authors of  \cite{discreteMultiple} studied a discrete-time single-queue multi-server system in which the service time and arrival interval of the data are subject to Bernoulli distribution. In particular, the precise distributions of AoI and PAoI are obtained under different preemption strategies and buffer settings in \cite{discreteMultiple}. %Compared to continuous-time systems, research on the information freshness of discrete-time systems is still in its early stages.
	%the works \cite{lisanfcfs,lisanber,lisanlcfs} extended  part of the above works from continuous-time systems to discrete-time systems, respectively. 

	Considering the frequent application of discrete-time systems in the IoT, modeling and studying the AoI of discrete-time systems can provide theoretical foundations for measuring information freshness in IoT systems. 
	Inspired by the findings in \cite{shuangduilie,mmmd} that the dual-queue transmission can improve the information freshness by sending status updates more frequently, we focus on the information freshness of parallel status updating in more practical discrete-time systems. In particular, we investigate the average AoI and PAoI of a discrete-time dual-queue system where two sensors independently sample the process following zero-wait (ZW) policy, generate status updates, and transmit them to a remote monitor via time-slotted channels. One of the two sensors is modeled as a ZW/Geo/1 queue, where the service time of a status update follows a geometric distribution. The other sensor is modeled as a ZW/D/1 queue, where the service time of a status update is deterministic. We refer to the dual-queue system as the Geo-D system. 		
	%  The work in  \cite{geogeo} considered  similar systems, where a queuing model consists of two ZW/Geo/1 queues, referred  to as Geo-Geo system. We additionally consider a simpler queue system consisting of two ZW/D/1 queues, referred to as the D-D system. We utilize the Geo-Geo and D-D systems for comparison with the Geo-D system in order to analyze the performance differences between dual-queue systems with different distributions.
	
	While the dual-queue transmission is promising in information freshness improvement, the randomness of service time of status updates in the parallel two queues leads to probable unordered arrival at the monitor, making the analysis of AoI and PAoI difficult. To address this problem, we adopt a graphical analysis method to derive the AoI and get analytic expressions for the average AoI and PAoI for the Geo-D system. Based on the results of AoI expressions, the connection between continuous-time and discrete-time systems is further explored.	The main contributions of this work are summarized as follows.
	%	 Specifically, we consider different cases of the system based on the succession of updates transmitted by the two queues, and define different cases as different states of the system. By characterizing the statistics of  terms related to updates, we obtain the expected values of the system AoI and PAoI in each state. Then, by analyzing the probability of occurrence of each state, we obtain the analytic expressions for the average AoI and PAoI of the system.  After obtaining analytic expressions for the average AoI and PAoI for Geo-D, we find that those are very similar to the results of \cite{mmmd}. 
	% The analytic expressions further  motivates us to briefly explore the connection of average AoI    between continuous-time systems and discrete-time systems.	The main contributions of this work are summarized as follows. 
	\begin{enumerate}
		\item The average AoI and PAoI are derived explicitly in the considered Geo-D discrete-time dual-queue system. In particular, we define the state of the system and present the probability of occurrence of each system in the steady state. By characterizing the statistics of the AoI-related terms in each state based on the graphical analysis method, the average AoI and PAoI are obtained in closed-form expressions.
		%\item We find that continuous-time systems can be regarded as the limiting case of discrete-time systems when the length of  time slot  approaches infinitesimal.
		%When the length of the time slot tends toward to infinitesimals,  the analytical expressions for the average AoI and PAoI in discrete-time systems can be extended to continuous-time systems through simple transformations.

		  \item It is shown that the average AoI and PAoI of the continuous-time counterpart system where the ZW/Geo/1 queue is replaced with a ZW/M/1 queue is just a limiting case of those of the considered Geo-D system, respectively. Similarly, the average AoI of the M-M  system, where the service time of both queues is exponential distributed, can be approached by a series of average AoIs of Geo-Geo systems. When the time slot length tends to infinitesimals,  the analytical expressions for the average AoI and PAoI in discrete-time systems can be extended to continuous-time systems through simple transformations.
		  
%		  By deriving the average AoI of the Geo-Geo system where the ZW/D/1 queue of the Geo-D system is replaced with another ZW/Geo/1 queue, we show that the average AoI of the M-M system where both queues are ZW/M/1 queues can be approached by a series of average AoI of the Geo-Geo system.

		% It is shown that the average AoI and PAoI of the continuous-time counterpart, M-D system is just a limit case of that of the considered discrete-time Geo-D system. Similarly, the average AoI of the M-M system where the service time of both queues is exp can be approximated  by a seriers of Geo-Geo system.
		\item Numerical results verify the effectiveness of the analysis. Compared to the ZW/Geo/1 and ZW/D/1  single-queue with the same service rates, our analysis reveals that the Geo-D dual-queue system can provide up to   39.74$\%$ and 19.65$\%$ average AoI reductions and up to   27.71$\%$  average PAoI reductions. In addition, the reduction ratio increases along with the decrease in the service rate. 
		In contrast to the dual-queue system, in terms of average AoI, both the Geo-D dual-queue system and the Geo-Geo dual-queue system outperform the D-D system at the same service rate for both queues. Moreover, Geo-D system performs better at low service rates and Geo-Geo system performs better at high service rates. In terms of average PAoI, while all the dual-queue system outperforms single-queue system, the randomness of service time contributes more AoI reduction than determinacy of service time in dual-queue system in most cases, which is different from what is known about the single queue system.
	\end{enumerate}
	
	The rest of this paper is organized as follows. Section \ref{II} introduces the system model. Section \ref{III} derives the analytical expressions for average AoI and PAoI. Section \ref{VI}  explores the connection between the discrete-time and continuous-time systems in terms of AoI. Simulation results are provided in Section \ref{IV} to evaluate the performance of the Geo-D dual-queue system. Finally, Section \ref{V} concludes this work.

	\section{System Model}\label{II}
 
	\begin{figure}[htp] 
		\centering
		\subfloat{\includegraphics[width=0.9\textwidth]{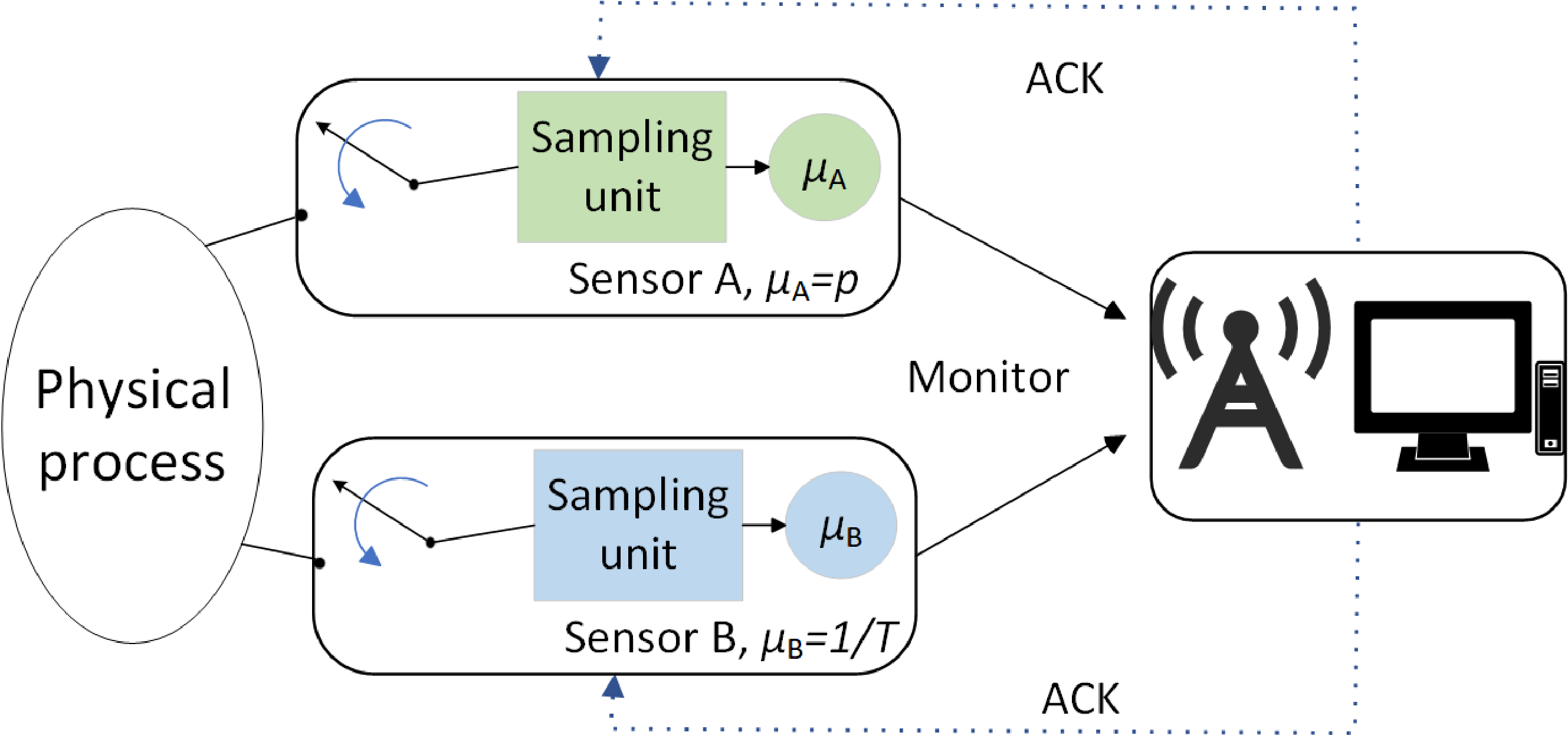}} 
		\caption{The considered dual-queue real-time monitoring system.}
		\label{xitong}
	\end{figure}

	\subsection{Dual-Queue Model}\label{firstanalyze}
	Let us consider a remote monitoring system consisting of two sensors and a monitor. The two sensors independently observe the same physical process and send update packets to the monitor to report its status. 	We assume that  time is divided into slots and that the system is operated slot by slot. Specifically, the duration of each slot is equal to the time required to transmit  one  update packet.
	At the beginning of each slot, the sensor sends  a status update. When the monitor receives a status update, it notifies the corresponding sensor for the generation of the next status update by using acknowledgment signaling (ACK). The transmission of an ACK signaling is considered instantaneous and error-free.  
	Since sending a status update,  the sensor would not generate a new status update until it receives an ACK. This update policy follows the ZW policy, under which the arrival interval time of the status update has been eliminated\cite{zerowait}.

	According to the operation policy, each sensor is modeled as a server, as shown in Fig. \ref{xitong}. 
	We name one of the sensors as sensor A and assume that its service time $S_{\text{A}}$ follows a geometric distribution with parameter $p\in [0,1]$. This models the case when the sensor updates the status through a wireless block fading channel. Accordingly, the service rate of sensor A can be expressed as $\mu_{\text A}=p$. The other sensor is named  sensor B and its service time $S_{\text{B}}$ is assumed to be deterministic with parameter $T\in\mathbb{N}^+$. This is motivated by the case when the sensor is connected to the monitor via a wired or dedicated channel for status updating. Accordingly, the service rate of sensor B  is expressed as $\mu_{\text {B}}=1/T$.   	
	Under the ZW policy, the two queues whose service time obeys geometric and deterministic distributions are referred to as a ZW/Geo/1 queue and a ZW/D/1 queue, respectively.
	These two parallel queues together form a dual-queue system, which is referred to as a Geo-D system.

	\subsection{Performance Metric}\label{perf_metr}
	The AoI at the monitor is denoted as the stochastic process
	\begin{equation}
		{\rm\Delta[\emph{t}]}=t-u(t), \,~~t\in\mathbb{Z},
	\end{equation}where $u(t)$ is the time-stamp of the most recently received
	status update corresponding to time slot $t$. According to the discrete Geo-D dual-queue model, the AoI evolution is expressed as
	\begin{equation}\label{dingyi}
		\rm\Delta[\emph{t}+1]=\left\{
		\begin{aligned}
			&\rm\Delta[\emph{t}]+1,\,\text{ no update has been served at \emph{t},}   \\
			&\min(t-\emph{$G_i$}+1,\rm\Delta[\emph{t}]+1),\,~~~~~\text{otherwise,}
		\end{aligned}
		\right.
	\end{equation}
	where $G_i$ is the  time that  $i$-th status update is generated. 
	Note from \eqref{dingyi} that, only if $\emph{t}-G_i < \rm\Delta[\emph{t}]$  holds \footnote{If the monitor receives two status updates at the same time and both status updates satisfy $\emph{t}-G_i < \rm\Delta[\emph{t}]$,  the monitor selects the one with a small AoI as a fresh status update and the other an obsolete status update. If $G_i$'s of the two status updates are the same, the monitor randomly selects one from the two as fresh and the other obsolete.}, would  $i$-th status update refresh the AoI at the monitor end to a small value equal to $t-G_i+1$. We call the status updates that satisfy  the condition ${t}-G_i < {\rm \Delta}[t]$ as fresh or valid status updates and the rest as obsolete or stale  status updates.

	Based on the service rate of the two sensors, the Geo-D system has two boundary cases in terms of AoI.
	For the first case, there is at least one server whose service rate reaches the upper limit, i.e., $\mu_{\text A}=1$ or $\mu_{\text {B}}=1$. In this case, the evolution of AoI of the Geo-D system is realized exclusively by the queue with a service rate equal to 1. Accordingly,  the values of AoI and PAoI would be a constant 2.
	For the second case, there exists one queue whose service rate is zero, i.e., $\mu_{\text A}=0$ or $\mu_{\text {B}}=0$ ($0<\mu_{\text A}+\mu_{\text B}<1$). In this case, the dual-queue system degenerates into a single-queue system, and the AoI and PAoI of the system  depend on the remaining queue. 
	In these two cases, the dual-queue system performs the same as one of the two parallel queues regarding information freshness. Therefore, in the sequel, we analyze
	the service rate within the boundary conditions, i.e., $0 < \mu_{\text {A}},\,\mu_{\text {B}} < 1$, in order to investigate how the two parallel queues affect the AoI evolution of the system, in  non-trivial cases. To  clarify, the AoI, the average AoI, and the average PAoI of a specific queue model $\psi$ are denoted by $\Delta_\psi$, $\overline{\Delta}_\psi$, and $\overline{\Delta}^{\text{p}}_\psi$, respectively. 
	% For cases where the service rates of a queue are paid attention, $\overline{\Delta}_\psi$, and $\overline{\Delta}^{\text{p}}_\psi$ are specified as functions of the service rates.

	%  And we find that when the service rate takes a boundary value, the dual-queue expressions of AoI and PAoI, become the corresponding single queue expressions.
	
	%		 [2/\mu-\overline{\Delta}(\mu,1/\mu)]/(2/\mu)=(2-\mu)(1-\mu)^\frac{1}{\mu}+(\mu-\frac{5}{2})(1-\mu)^\frac{2 }{\mu}, 
	
	%	\begin{remark}When $\mu_{\text A}=\mu_{\text {B}}=\mu$, ${ \overline{\rm \Delta}_{\rm ZW/Geo/1 \rm}}=2/\mu$ and ${ \overline{\rm \Delta}_{\rm ZW/D/1 \rm}}=3/(2\mu)+1/2$. Both the average PAoI of  ZW/Geo/1 queue system and that of  ZW/D/1 queue system is $2/\mu$.  Compared with ZW/Geo/1 queue system, the reduction ratio of average AoI of Geo-D  system is $[2/\mu-\overline{\Delta}(\mu,1/\mu)]/(2/\mu)=(2-\mu)(1-\mu)^\frac{1}{\mu}+(\mu-\frac{5}{2})(1-\mu)^\frac{2 }{\mu}$.   Compared with ZW/D/1 system, the reduction ratio of average AoI  is $[3/(2\mu)+1/2-\overline{A}(\mu,1/\mu)]/((3/(2\mu)+1/2)=[{\mu-1 +4(2-\mu)(1-\mu)^\frac{1}{\mu}+2(2\mu-5)(1-\mu)^\frac{2}{\mu} }]/({\mu+3})$.  Compared with both ZW/D/1 queue system and ZW/Geo/1 queue system, the reduction ratio of the average PAoI of Geo-D system, is $(2/\mu-\overline{A}(\mu,1/\mu))/(2/\mu)=[{(\mu-2)(1-\mu)^\frac{2}{\mu}+(1-\mu)^{\frac{1}{\mu}-1}}(3\mu-1-\mu^2)]/[{2(1-\mu)^\frac{2}{\mu}+(1-\mu)^{\frac{1}{\mu}-1}+1}]$. Please refer to Numerical Results section for a more detailed discussion on the potential reduction.
		%\end{remark}
		\section{ AoI of the Geo-D System }\label{III}
		%		In this section, we derive the AoI of the Geo-D system using graphic-based analysis. 
		%		Due to the randomness of the  service time of sensor A, status updates successfully served by both sensor A and sensor B would probably be obsolete. This randomness makes difficult to analyze the AoI based on each status update arrived at the monitor end. By carefully examining the possible AoI evolution at the monitor end, we consider  to classify the system into different states based on arrival pattern of updates at the monitor. Then we  analyze and calculate the  AoI and PAoI  of each state, and finally derive the closed form expressions for the average AoI and PAoI of the Geo-D system.  
		In this section, we derive the AoI of the Geo-D system based on the graphical analysis method. 
		Due to the randomness of the service time of sensor A, status updates successfully served by both sensors A and B would probably be obsolete. This randomness makes it challenging to  analyze the AoI. We consider classifying the system into different states by examining the possible AoI evolution pattern at the monitor end.
  % we consider classifying the system into different states. 
  Then, we analyze and calculate the  AoI and PAoI of each state and derive the closed-form expressions for the average AoI and PAoI of the Geo-D system. Finally, we discuss the improvement of the information freshness brought by the redundancy of each sensor in the Geo-D system.

		%Therefore, in order to analyze whether the  status updates received by the monitor at the receiver node are stale or fresh, we need analyze all the situations when the monitor receives  status updates.
		%Depending on the number of  status updates in each sensor transmitted, the evolution of AoI in the system  is divided into a finite number of situations. We  call these situations as system states.  We  analyze and calculate the  AoI and PAoI  of each state of Geo-D system, and finally derive the closed form expressions for the average AoI and PAoI.
		\subsection{System State Analysis}\label{state}
		%	Since the service time of sensor B in the Geo-D system is deterministic, we can evaluate the
		%	average AoI and PAoI of the Geo-D system by analyzing the staircase-like  AoI waveform during a service
		%	period $T$ of sensor B. 
		
		%		In order to show the AoI evolution of each state  clearly, the AoI paths of each state is shown in Fig \ref{zhuangtai}. 

		It is valuable to note that the service time of each status update from sensor B is deterministic, i.e., $T$ slots. Accordingly, the AoI evolution of 
		status updates served by sensor B presents a periodical staircase-like waveform. Fig. \ref{zhuangtai}  presents some typical examples of the AoI evolution path in two consecutive periods.
		The dotted and  solid curves in  Fig. \ref{zhuangtai} represent the evolution of AoI of status updates from sensor A and sensor B, respectively. The shaded area represents the evolution of the AoI of the monitor in the current period.
		
	The AoI of the system, which is the minimum of the AoI of updates coming from sensor A and that coming from sensor B, forms a random staircase-like AoI  and is affected by the number of updates successfully served by sensor A in the current period and the previous period (cf. Fig. \ref{zhuangtai}). This motivates us to define the state of the system based on the number of successfully served status updates coming from sensor A in each period of $T$ slots. Specifically, let $K$ and $N$ denote the number of state updates served by sensor A in the previous period and the current period, respectively. 
		%		
		%		Given that the service time of sensor B in the Geo-D system is deterministic, we can assess the average AoI and PAoI of the Geo-D system by analyzing the staircase-like AoI waveform during a period $T$ of sensor B.  Moreover, the area covered by the staircase-like  AoI waveform is determined by the number of  status updates sent by sensor A in the current period and the previous period, shown as Fig. \ref{zhuangtai}.
		%		Let $K$ and $N$ denote the number of  state updates served by sensor A in the previous period and the current period, respectively.
		%		According to $K$ and $N$,  we can categorize the system state into different cases. We use the notation $(k,n)$ to represent the system state of Geo-D system, where $k\in\mathbb{N}$ and $n\in\mathbb{N}$ represent the number of   status updates transmitted by sensor A in the previous period and the current period, respectively. By definition, $0\leq k,n \leq T$. 	
		According to the values of $K$ and $N$, we say the Geo-D system is in state $(k,n)$.
		%		we  categorize the state of the system into different cases.  
		%		We use the tuple $(k,n)$ to represent that  $K=k$ and $N=n$, where $k\in\mathbb{N}$ and $n\in\mathbb{N}$, and, 
		%		we say the Geo-D system  is in state $(k,n)$.
		Since the service time of sensor A follows a geometric distribution with parameter $p$, the number of status updates served by sensor A during $T$ slots follows a binomial distribution, i.e., 	$K \sim { \textsf{B}} (T,p)$ and  $N \sim { \textsf{B}} (T,p)$. The probability of the Geo-D system in a specific state $(k, n)$ can be expressed as 		
		\begin{equation}\label{2}
			\mathbb{P}_{K,N}(k,n)=\binom{T}{k}p^k(1-p)^{T-k}\binom{T}{n}p^n(1-p)^{T-n}.
		\end{equation}\par 	
		
		%		Now let us specify each state and explain how the number of status updates sent by A in the previous and current periods affects the number of valid status updates received by the monitor in the current period. With these explanations, we shall clarify $V(k,n)$,  the average number of valid status updates received by the monitor in state $(k,n)$.
		
		Now, let us specify each state and explain how the number of status updates sent by sensor A in the previous and current periods affects the number of valid status updates received by the monitor in the current period, i.e., $V(k,n)$. 
		
		%		The division of these states is explained as follows.
		%To keep clarity, a  schematic diagram of AoI paths for each state  is shown in Fig. \ref{zhuangtai}.
		\begin{figure*}[htp] 
			\subfloat[state $(0,0)$\label{new00}]{\includegraphics[width=0.32\textwidth]{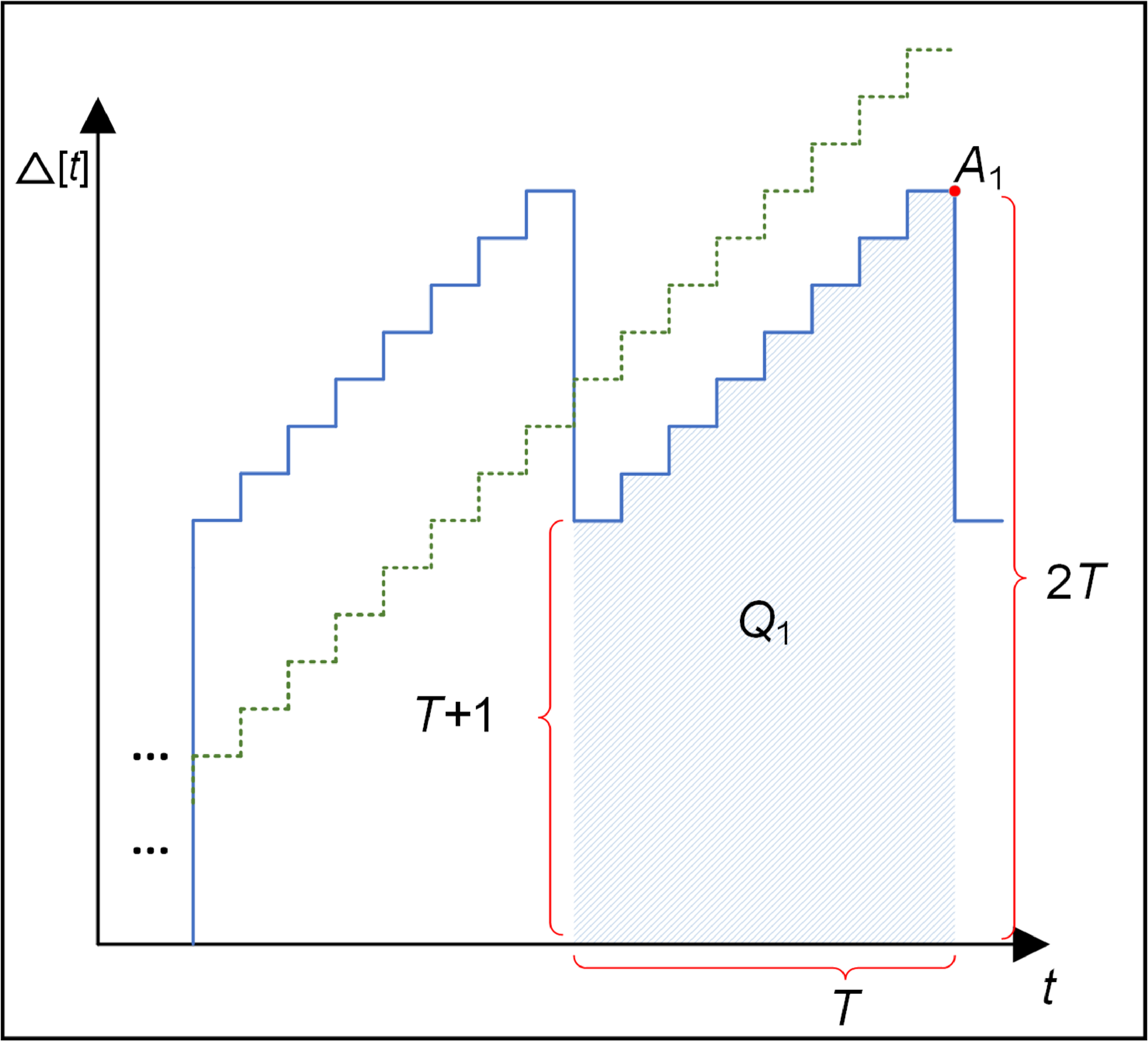}} 
			\hfill 	
			\subfloat[state $(0,1)$\label{new01}]{\includegraphics[width=0.32\textwidth]{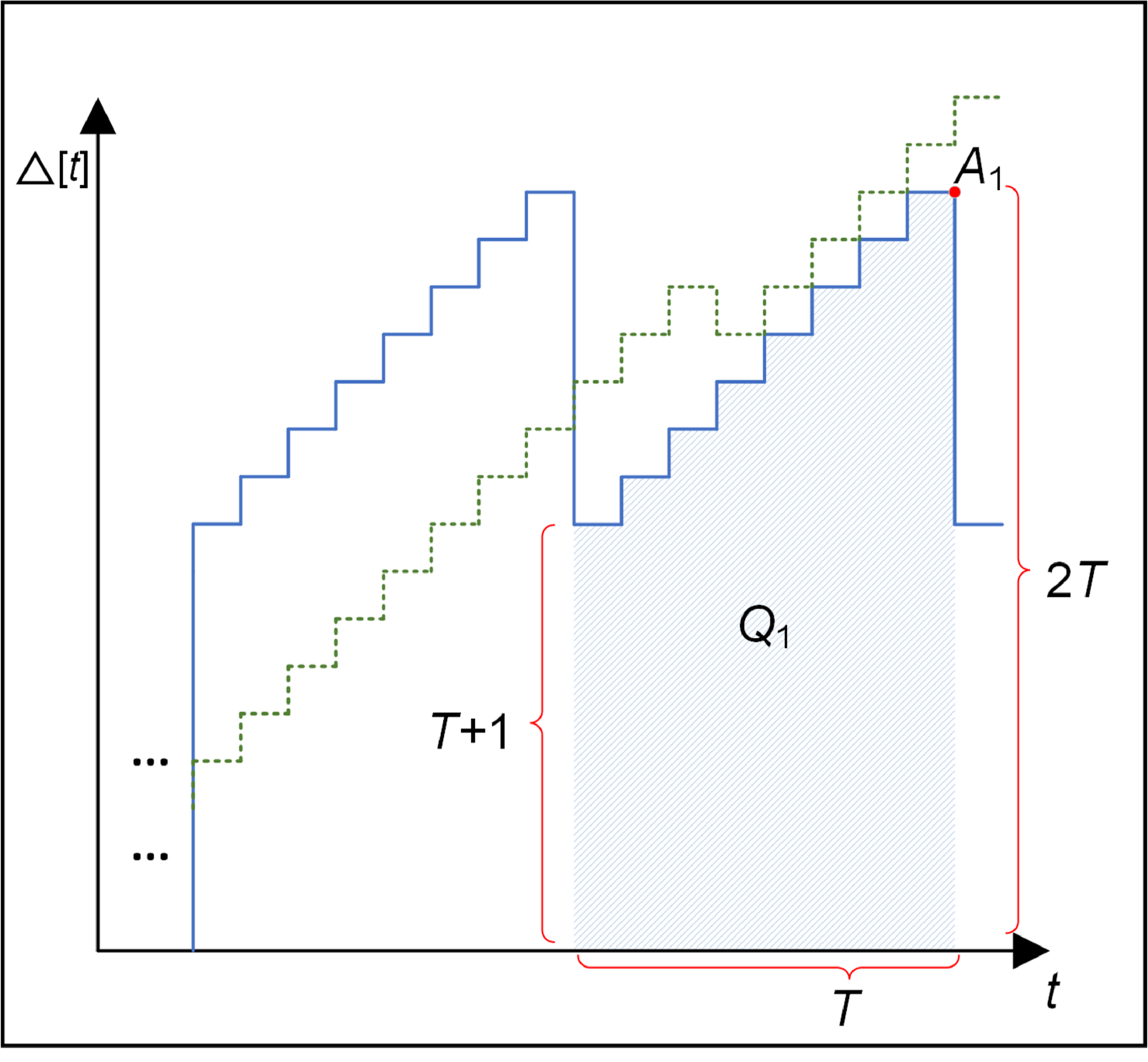}}  
			\hfill	
			\subfloat[state $(0,n),~2\leq n\leq T$\label{new0n}]{\includegraphics[width=0.32\textwidth]{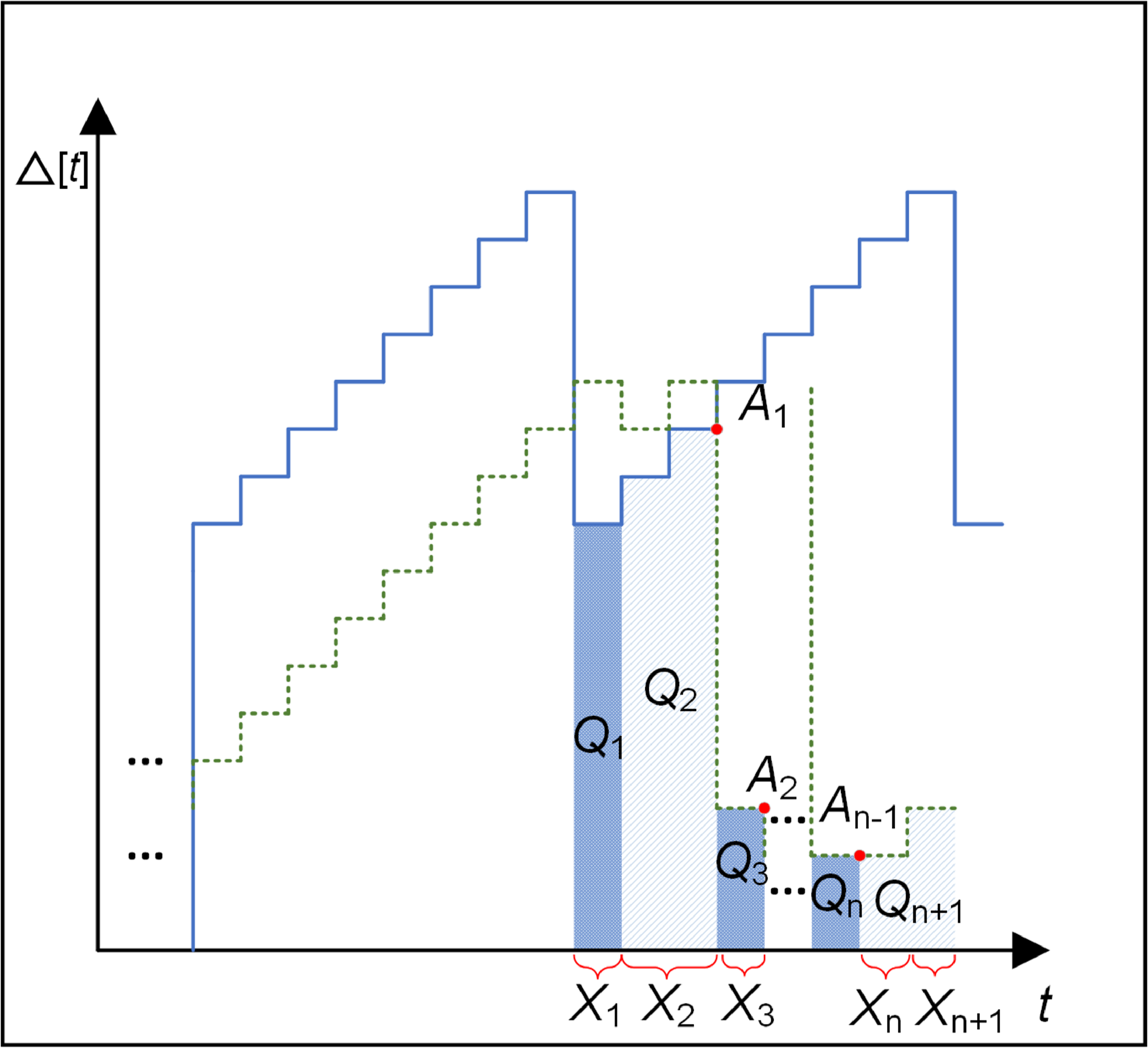}}
			\newline
			\subfloat[state $(k,0),~1\leq k\leq T$\label{newk0}]{\includegraphics[width=0.32\textwidth]{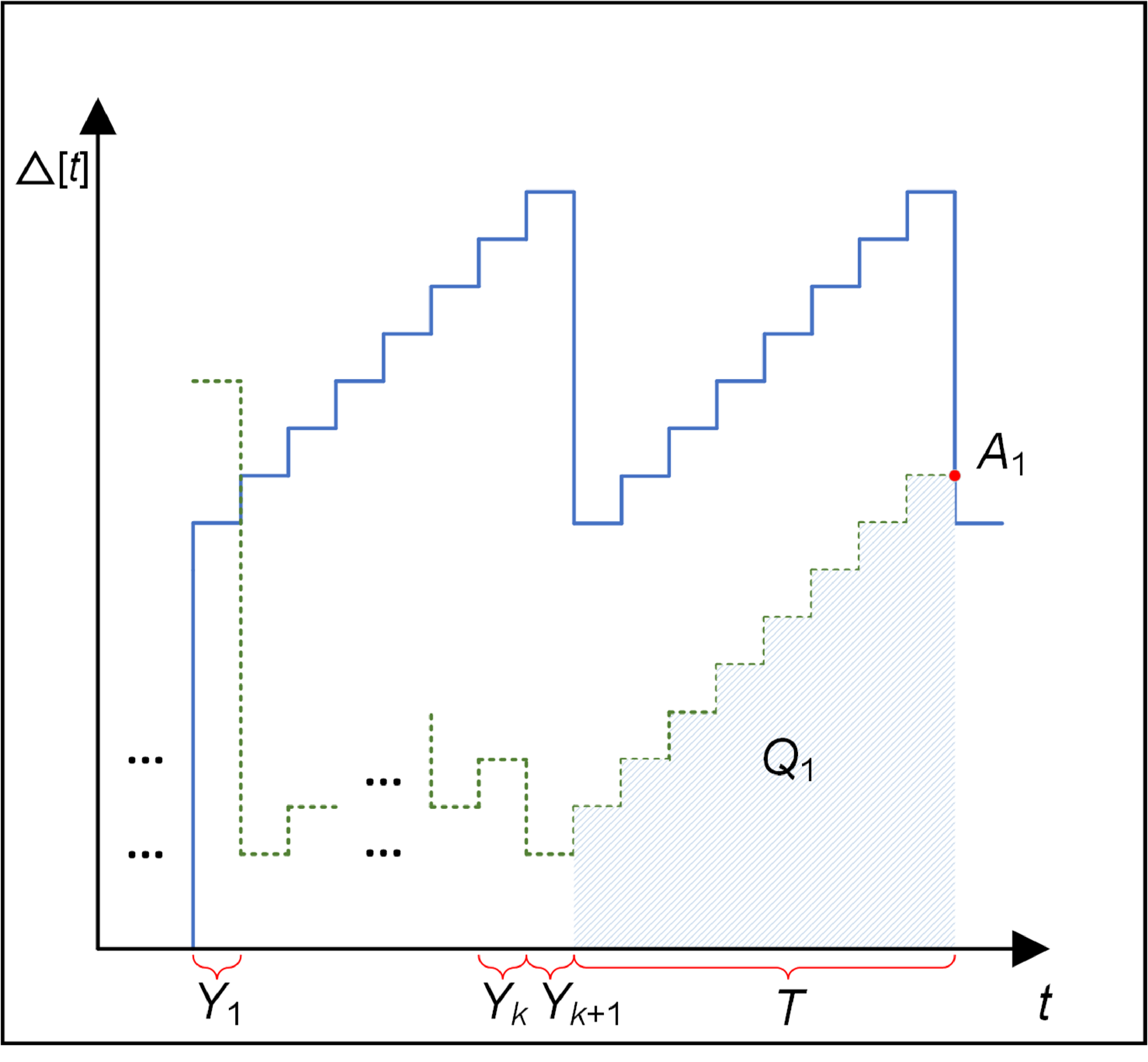}}
			\hfill
			\subfloat[state $(k,1),~1\leq k\leq T$\label{newk1}]{\includegraphics[width=0.32\textwidth]{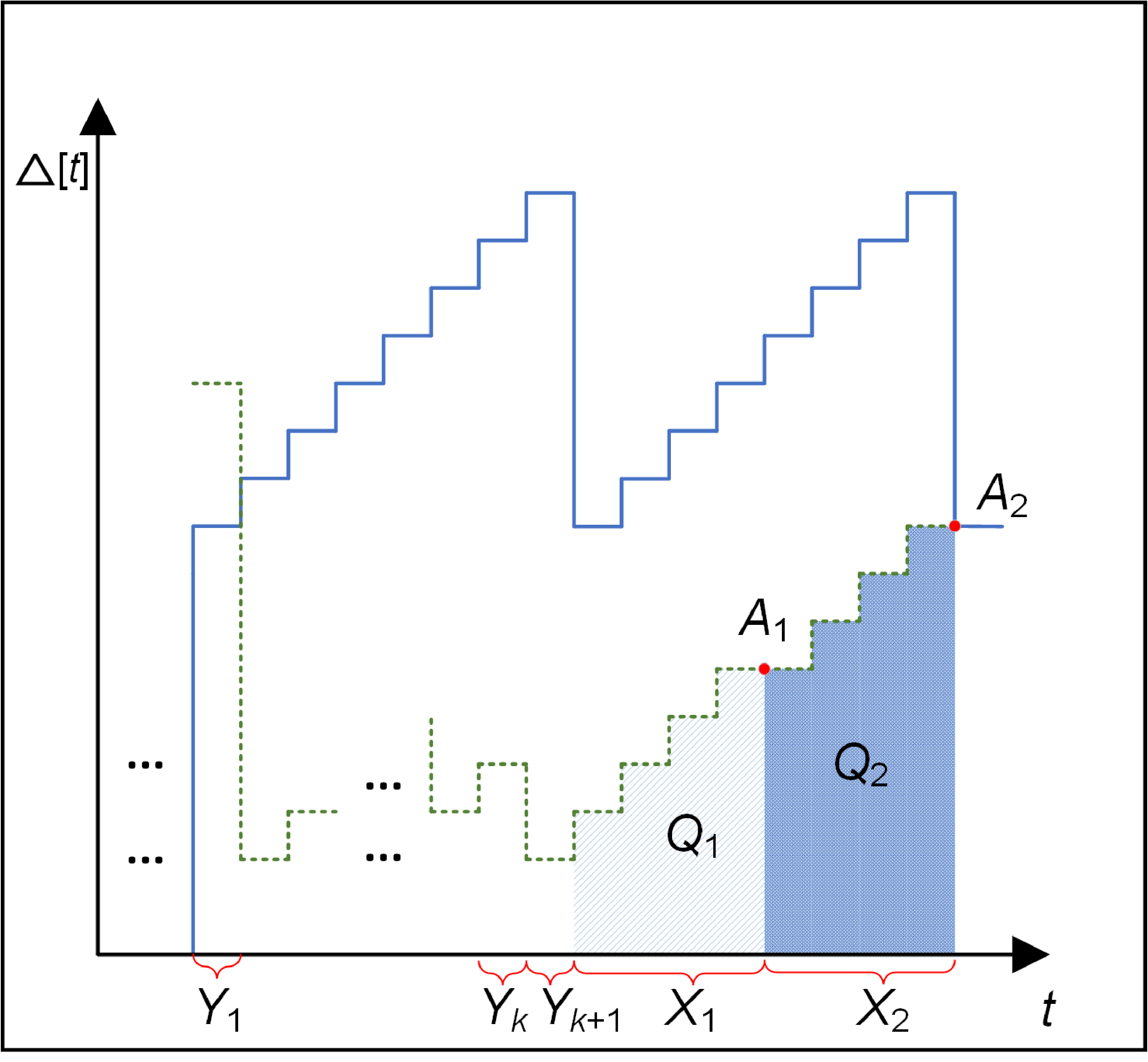}}
			\hfill	
			\subfloat[state $(k,n),~1\leq k\leq T,~2\leq n\leq T$\label{newkn}]{\includegraphics[width=0.32\textwidth]{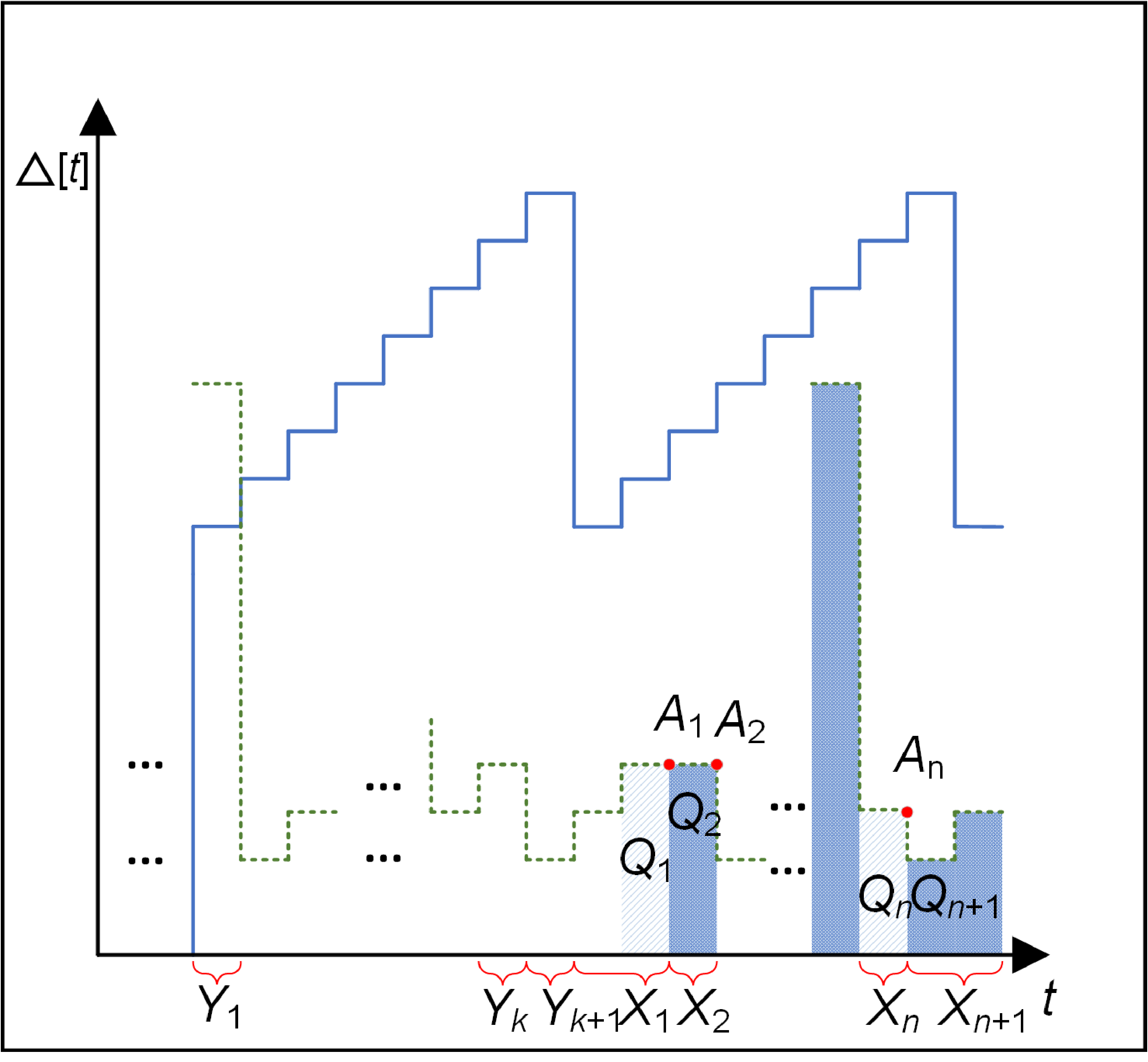}}
			\caption{Evolution examples of  AoI for the Geo-D system in each state.}
			\label{zhuangtai}
		\end{figure*}
		\subsubsection{State $(0,0)$} During the previous  period and the current  period, no  status update has been successfully served by sensor A. Note from Fig. \ref{zhuangtai}\subref{new00} that the monitor would  receive one valid status update from sensor B at the $T$-th slot of the current period. Then, $V(0,0)=1$.
		\subsubsection{State $(0,1)$}
		As shown in Fig. \ref{zhuangtai}\subref{new01}, there is no  status update and only one status update successfully served by Sensor A in the previous and the current periods, respectively.  
		Since  no fresh status update was successfully served during the previous period, the update  served by sensor A would be staler than the freshest status update at the monitor, i.e., the one successfully served by sensor B at the end  of the previous period. Hence, the monitor only receives one valid status update from sensor B at the end of the current period, i.e., $V(0,1)=1$.
		\subsubsection{State $(0,n)$, $2\leq n\leq T$}
		Similar to state $(0,1)$, the first status update sent by sensor A in the current period is obsolete. 
		% and does not refresh the AoI of the monitor  since no status update is successfully served in the previous  period.
		However, the other $n-1$ status updates that arrived in the current period contribute to the refreshing of the AoI of the monitor, as shown in  Fig. \ref{zhuangtai}\subref{new0n}.
		Hence, the status update from sensor B served at $T$-th slot is  stale and does not refresh the AoI of the monitor. Therefore, the monitor  receives $n-1$ valid status updates in the current period. That is, $V(0,n)=n-1$.
		\subsubsection{State $(k,0)$, $1\leq k\leq T$}Sensor A sent $k$  status updates in the previous period, yet it does not send any new  status update in the current period (cf. Fig.\ref{zhuangtai}\subref{newk0}). In this case, the monitor only receives one valid status update from sensor B in the period of interest. That is, $V(k,0)=1$.  
		%		The instantaneous staircase-like AoI waveform for the Geo-D system is shown in Fig. \ref{zhuangtai}\subref{newk0}. 
		\subsubsection{State $(k,1)$, $1\leq k\leq T$}
		Different from state $(k,0)$, there is a successfully served status update of sensor A arrived at the monitor in the current period, as shown in Fig. \ref{zhuangtai}\subref{newk1}.		
		In particular, there is a probability of $(1/T)$ that sensor A completes the service of the only successfully served update at the $T$-th  slot, i.e., simultaneously with that successfully served by sensor B. By definition,   the one coming from sensor A is valid, while the one from sensor B is obsolete.			
		Therefore, in this case, the monitor receives only one valid status update from sensor A, i.e., $V(k,1)=1$. It is valuable to note that in the rest cases, sensor A completes its service before sensor B. Hence the monitor receives two valid  status updates, one from sensor A and the other from sensor B with probability $(1-1/T)$, i.e., $V(k,1)=2$. 
		A typical example for the case where $V(k,1)=2$ can be found in Fig. \ref{zhuangtai}\subref{newk1}.
		\subsubsection{State $(k,n)$, $1\leq k\leq T$, $ 2\leq n\leq T$}
		As shown in Fig. \ref{zhuangtai}\subref{newkn}, in the current period, all $n$ status updates from sensor A are valid, while the one from sensor B is obsolete. Hence, the monitor receives  $n$ valid status updates in total. That is, $V(k,n)=n$.		
		\subsection{Average AoI and PAoI of the Geo-D System} \label{calculation}
		Based on the system states, we can compute the AoI and PAoI of each state and those of the Geo-D system. To characterize the AoI evolution in each state, let us define the following random variables.
		Let $X_i\,(i=1, 2, \cdots, n)$ and  $Y_j\,(j=1, 2, \cdots, k)$ denote the service time of the $i$-th and $j$-th  status update from sensor A  within the current and the previous period for state $(k,n)$, respectively (cf. Fig. \ref{zhuangtai}). %Specifically, when $k=n=0$, we define $X_0=x_0=0$ and $Y_0=y_0=0$.
		In the current  and previous  period , the remaining time after sensor A has successfully sent $n$ and $k$   status updates are denoted by $X_{n+1}=T-\sum_{i=1}^{n}X_i$ and $Y_{k+1}=T-\sum_{j=1}^{k}Y_j$, respectively. With these definitions, the expected sum of PAoI and  that of AoI in  each state $(k,n)$ can be calculated by the following equations (also cf. Fig. \ref{zhuangtai}).
		%		\begin{equation}\label{dingyi}
			%			\mathbb{E}[A|(k,n)]=\left\{
			%			\begin{aligned}
				%				&\frac{1}{T}\mathbb{E}[A|(k,n)]+({1-\frac{1}{T}})\mathbb{E}[A|(k,n)],  1 \leq k \leq T, n=0  \\
				%				&\min(t-\emph{$G_i$}+1,\rm\Delta[\emph{t}]+1),\,~~~~~\text{otherwise,}
				%			\end{aligned}
			%			\right.
			%		\end{equation}
		\begin{align}
			\mathbb{E}[A|(k,1)]=&\frac{1}{T}\mathbb{E}[A_1|(k,1)]+\notag\\
			&({1\!-\!\frac{1}{T}})\sum_{i=1}^{2}\mathbb{E}[A_i|(k,1)], \quad1\leq k\leq T,\label{EAkn}\\
			\mathbb{E}[A|(k,n)]=&\sum_{i=1}^{V(k,n)}\mathbb{E}[A_{i}(k,n)], ~~~~ n\neq 1 \text{~or~} k=0,\\
			\mathbb{E}[Q|(k,n)]=&\sum_{i=1}^{n+1}\mathbb{E}[Q_{i}(k,n)], 
		\end{align}
		%		\begin{align}\label{EAkn}
			%			\mathbb{E}[A|(k,n)]=&\sum_{i=1}^{V(k,n)}\mathbb{E}[A_{i}(k,n)],\\
			%			\mathbb{E}[Q|(k,n)]=&\sum_{i=1}^{n+1}\mathbb{E}[Q_{i}(k,n)],
			%		\end{align}
		where  $\mathbb{E}[A_{i}(k,n)]$ and $\mathbb{E}[Q_{i}(k,n)]$ respectively represent the $i$-th component AoI peak and area of  the right angled trapezoid with staircase-like hypotenuse (cf. the shaded areas marked with $A_i$ and $Q_i$ in Fig. \ref{zhuangtai}) in state $(k,n)$.
		In  particular, when analyzing state $(k,n)$, $A_{i}(k,n) $ and $ Q_{i}(k,n)  $ are simply written as $A_i$ and $Q_i$.
		
  The average AoI of the Geo-D system is obtained by summing the expectations for each state.
		Note that the average AoI under each state is the staircase-like waveform coverage area $Q$ divided by the period time $T$.
		Then, the average AoI of the Geo-D system can be  obtained as
		\begin{align}       %双栏准备
			% {\overline\Delta_{\text{Geo-D}}}&=\sum_{k=0}^{T}\sum_{n=0}^{T} \mathbb{P}_{K,N}(k,n)  \cdot  {\mathbb{E}[\Delta|(k,n)]} \nonumber \\
			{\overline\Delta_{\text{Geo-D}}}&=\frac{1}{T}\sum_{k=0}^{T}\sum_{n=0}^{T} \mathbb{P}_{K,N}(k,n)  \cdot  {\mathbb{E}[Q|(k,n)]}.
			\label{AoI}
		\end{align}
		% Accordingly, the sum of the expectation  of  PAoI  for the Geo-D system is calculated as
		% \begin{equation}
			% 	\label{ea}
			% 	\mathbb{E}[A]=\sum_{k=0}^{T}\sum_{n=0}^{T} 	\mathbb{P}_{K,N}(k,n)  \mathbb{E}[A|(k,n)].
			% \end{equation}
		Note that the average number of valid status updates of the Geo-D system  in a period can be calculated by
		\begin{equation}
			\label{en}
			\mathbb{E}[V]=\sum_{k=0}^{T}\sum_{n=0}^{T}\mathbb{P}_{K,N}(k,n) \mathbb{E}[V(k,n)],       
		\end{equation} 
		where % $\mathbb{E}[N|(k,n)]$ represents the  number of valid status updates in state $(k,n)$ and 
		$V(k,n)$ is given by the specific explanation of each state in Sec. III-A. Accordingly, the average PAoI of the Geo-D system can be calculated through dividing the sum of PAoI by the average number of valid status updates, i.e.,
		\begin{equation}
			{\overline{\Delta}^{\text{p}}_{\text{Geo-D}}}= \frac{\mathbb{E}[A]}{\mathbb{E}[V]}=\frac{\sum_{k=0}^{T}\sum_{n=0}^{T} 	\mathbb{P}_{K,N}(k,n)  \mathbb{E}[A|(k,n)]}{\sum_{k=0}^{T}\sum_{n=0}^{T}\mathbb{P}_{K,N}(k,n) \mathbb{E}[V(k,n)]}.\label{PAoI}
		\end{equation}\par
		% The average AoI of the Geo-D system is obtained by summing the expectations for each state.
		% Note that the average AoI under each state is the staircase-like waveform coverage area $Q$ divided by the period time $T$.
		% Then, the average AoI of the Geo-D system can be  obtained as
		% \begin{align}       %双栏准备
			% 	{\overline\Delta_{\text{Geo-D}}}&=\sum_{k=0}^{T}\sum_{n=0}^{T} \mathbb{P}_{K,N}(k,n)  \cdot  {\mathbb{E}[\Delta|(k,n)]} \nonumber \\
			% 	&=\frac{1}{T}\sum_{k=0}^{T}\sum_{n=0}^{T} \mathbb{P}_{K,N}(k,n)  \cdot  {\mathbb{E}[Q|(k,n)]}.

			% \end{align}\par
		
		According to \eqref{EAkn}-\eqref{PAoI}, to derive average AoI and PAoI, it only remains to compute the expectations of $A_i$ and $Q_i$ for each state. To this end, let us introduce the following lemma.
		\begin{lemma} \label{ol1}
			%		\textit{Let $X_1$,$X_2$,$\cdots$,$X_n$ be n non-negative independent  discrete random variables satisfying $T-\sum_{i=1}^{n}X_n \geq 0$ , $X_i \in \mathbb{N^+}$ and $n\leq T$, $n,T\in \mathbb{N^+}$. Then the following equation holds}
			Consider two integers $n,~T\in \mathbb{N^+}$ satisfying $1\leq n \leq T$. It holds that
			%								Consider there are $n\leq T$ positive integers $x_1, x_2,\cdots, x_n$ that satisfy $T-\sum_{i=1}^{n}x_n \geq 0$, where  $n,T\in \mathbb{N^+}$. Specially, let $x_0=0$. Then the following equations holds
			\begin{align}
				%&\sum_{x_1=1}^{b_1}\cdots\sum_{x_i=1}^{b_i}\cdots\sum_{x_n=1}^{b_n} 1=\binom{b_1+n-1}{n}, \label{base1} \\
				\sum_{x_i=1}^{b_i}\sum_{x_{i+1}=1}^{b_{i+1}}\cdots\sum_{x_n=1}^{b_n} 1=&\binom{b_i+n-i}{n-i+1}, \label{base1} \\
				\sum_{x_1=1}^{b_1}\sum_{x_2=1}^{b_2}\cdots\cdots\sum_{x_n=1}^{b_n}x_i=&\binom{T+1}{n+1}, \label{lemma1} \\
				\sum_{x_1=1}^{b_1}\sum_{x_2=1}^{b_2}\cdots\cdots\sum_{x_n=1}^{b_n}x_i^2 =&\binom{T+2}{n+2}+\binom{T+1}{n+2}, \label{lemma2} \\
				\sum_{x_1=1}^{b_1}\sum_{x_2=1}^{b_2}\cdots\cdots\sum_{x_n=1}^{b_n}x_ix_j=&\binom{T+2}{n+2}, \label{lemma3}
			\end{align}
			where $1\leq i,j \leq n$, $i\neq j$,  and
			\begin{align}\label{bi}
				b_i:=T-(n-i)-\sum_{l=1}^{i-1}x_l, ~~~~ 1\leq i\leq n.
			\end{align}
		\end{lemma}
		\begin{proof}
			Please refer to Appendix \ref{A} for the proof.
		\end{proof}
		
		Based  on  Lemma \ref{ol1}, we have the following result.
		\begin{Theorem}\label{Pro1}
			In the Geo-D dual-queue parallel transmission system, where $0\, \leq\, \mu_{\text A}=p\leq1$, $q:=1-p$, and $\mu_{\text {B}}=1/T$, $T\in \mathbb{N^+}$, the average PAoI and average AoI are, respectively,
			\begin{align}\label{expaoi}
				{\overline{\Delta}^{\text p}_{\text{Geo-D}}}=&\frac{1}{q^{2T\!-\!1}\!+\!q^{T\!-\!1}(T\!-\!1)p\!+\!Tp}\bigg[2T\!+\!q^{2T\!-\!1}\Big[\frac{2}{p}\!-\!(\frac{p}{2}\nonumber \\
				&\!-\!2)(T\!-\!1)\Big]\!+\!q^{T\!-\!1}\Big(2\!-\!\frac{2}{p}\!-\!\frac{p}{2}\!-\!\frac{Tp}{2}\!+\!T^2p\Big)   \bigg], \\
				\label{exaoi}
				{ \overline{\rm \Delta}_\text{Geo-D}}=&\frac{2}{p}\!+\!\frac{q^T}{p}\Big(\!-\!1\!+\!\frac{2}{T}\!-\!\frac{3}{pT}\Big)\!+\!\frac{q^{2T}}{p}\Big(2\!-\!\frac{2}{T}\!+\!\frac{3}{pT}\Big).
			\end{align}
		\end{Theorem}

		%		\begin{proof}
			%			The proof is provided in Section \ref{III}.
			%		\end{proof}
		%	\end{proposition}
	%	\begin{proposition}\label{Pro2}
		%		In the Geo-D dual-queue parallel transmission system, where $0\, \leq\, \mu_{\text A}=p\leq1$, $q:=1-p$, and $\mu_{\text {B}}=1/T$, $T\in \mathbb{N^+}$, the average AoI is
		%		\begin{equation}\label{exaoi}
			%			{ \overline{\rm \Delta}(p,T)}=\frac{2}{p}+\frac{q^T}{p}\Big(-1+\frac{2}{T}-\frac{3}{pT}\Big)+\frac{q^{2T}}{p}\Big(2-\frac{2}{T}+\frac{3}{pT}\Big).
			%		\end{equation}
		%		\begin{proof}
			%			The proof is provided in Section \ref{III}.
			%		\end{proof}

		\begin{proof}
			
State (0,0): During the previous service period and the current service period of sensor
B, no update from sensor A arrives at the monitor. Figure \ref{zhuangtai} \subref{new00} note that the monitor only received an update from sensor B during the current service
period. As shown in Figure \ref{zhuangtai} \subref{new00}.
Therefore, the PAoI is
\begin{equation}
	\mathbb{E}[A|0,0]=T+T=2T.
\end{equation}
Moreover, the area covered by the AoI waveform at the monitor, i.e., the shaded part in 
Figure \ref{zhuangtai} \subref{new00} is 
\begin{equation}
	\mathbb{E}[Q|0,0]=\frac{(T+1+T+T)T}{2}=\frac{3T^2+T}{2}.
\end{equation}
State (0,1): Sensor A did not sent an update in the previous service period but transmits
one update in the current period, as shown in 
Figure \ref{zhuangtai} \subref{new01}.
Since there was no update during the last service period, the updated AoI of sensor A must not be any better than the AoI of sensor B.
Therefore, such an update from sensor A  does not refresh the AoI on the monitor.
Consequently, the sum of the average PAoI in this case is
\begin{equation}
	\mathbb{E}[A|0,1]=T+T=2T.
\end{equation}
The average area under AoI curve of the Geo-D system in state (0, 1) is given by
\begin{equation}
	\mathbb{E}[Q|0,1]=\frac{(T+1+T+T)T}{2}=\frac{3T^2+T}{2}.
\end{equation}

State (0,n): Sensor A transmits $n\,(n \geq2)$ updates in the current service period since no updates were performed in the previous service period. Similar to state (0,1), the first update sent by sensor A in the current service period is stale and does not make the system update since no update was made in the previous service period. The system totally receives $n-1$ valid updates in the current service period (the update of sensor B, in this case, is also stale).
As shown in  Figure \ref{zhuangtai} \subref{new0n}. Let $X_i\,(i = 1, 2, ..., n)$ denote the service time of the $ith$
update. Within a period $T$, the remaining time after sensor A successfully sent n updates is $X_{n+1}$.
According to the properties of Bernoulli process, the joint probability distribution function of
$(X_1, ..., X_n)$ is
\begin{equation}
	P_{X_1,...,X_n}(x_1,...,x_n|N(T)=n)=\frac{p(1-p)^{x_1-1}...p(1-p)^{x_{n}-1}(1-p)^{T-\sum_{i=1}^nx_i}}{\binom{T}{n}p^n(1-p)^{T-n}},
\end{equation}
Reduced to:
\begin{equation}
	P_{X_1,...,X_n}(x_1...,x_n|N(T)=n)=\frac{1}{\binom{T}{n}}.
\end{equation}
The first PAoI of the monitor is  $A_1=T+X_1+X_2$ and $A_i=X_i+X_{i+1}\,(i\geq2)$. It can be obtained that the expectation of the first PAoI is:
\begin{align}
	\mathbb{E}(A_1) &= \sum_{x_1=1}^{T-n+1}...\sum_{x_i=1}^{T-(n-i)-\sum_{m=1}^{i-1}x_m}...\sum_{x_n=1}^{T-\sum_{i=1}^{n-1}x_i}   P_{X_1,...,X_n}(x_1,...,x_n|N(T)=n)\nonumber\\
	&\qquad\cdot(T+x_1+x_2)  \nonumber \\
	&=\frac{1}{\tbinom{T}{n}}\sum_{x_1=1}^{T-n+1}...\sum_{x_i=1}^{T-(n-i)-\sum_{m=1}^{i-1}x_m}...\sum_{x_n=1}^{T-\sum_{i=1}^{n-1}x_i}(T+x_1+x_2) .
	\label{ea1}
\end{align}
To compute the above expressions, we introduce the following lemma.

Using Lemma 1, we can rewrite \eqref{ea1} as
\begin{equation}
	\mathbb{E}(A_1)=\frac{2(T+1)}{n+1}+T.
\end{equation}
For the sake of brevity, we define the upper limit of the i-th summation of $X_i$ as:
\begin{align}
	b_1=T-n+1,\,b_i=T-(n-i)-\sum_{m=1}^{i-1}x_m,\, i=2,3....n.
\end{align}
The mean of the second PAoI $A_2$ can be expressed as
\begin{align}
	\mathbb{E}(A_2) &= \sum_{x_1=1}^{b_1}\sum_{x_2=1}^{b_2}...\sum_{x_i=1}^{b_i}...\sum_{x_n=1}^{b_n}   P_{X_1,...,X_n}(x_1,...,x_n|N(T)=n)\cdot(x_2+x_3)\nonumber\\
	&=\frac{1}{\binom{T}{n}}\sum_{x_1=1}^{b_1}...\sum_{x_i=1}^{b_i}...\sum_{x_n=1}^{b_n}(x_2+x_3)=\frac{2(T+1)}{n+1} .
\end{align}
Similarly, we have
\begin{equation}
	\mathbb{E}(A_2)=\mathbb{E}(A_3)=...=\mathbb{E}(A_{n-1}) =\frac{2(T+1)}{n+1} .
\end{equation}
Therefore, the sum of the average PAoI for the Geo-D system in state (0, n) is
\begin{equation}
	\mathbb{E}[A|0,n]=\sum_{i=1}^{n-1}\mathbb{E}(A_i)=T+\frac{2(n-1)(T+1)}{n+1}.
\end{equation}
The average of the shaded trapezoid areas $Q_1$, $Q_2$, and $Q_3$ in  Figure \ref{zhuangtai} \subref{new0n}  are given respectively by the following
\begin{align}
	\mathbb{E}(Q_1)=&\sum_{x_1=1}^{b_1}\sum_{x_2=1}^{b_2}...\sum_{x_n=1}^{b_n}   P_{X_1,...,X_n}(x_1...,x_n|N(T)=n)\frac{(2T+1+x_1)x_1}{2}, \\
	\mathbb{E}(Q_2)=&\sum_{x_1=1}^{b_1}\sum_{x_2=1}^{b_2}...\sum_{x_n=1}^{b_n}P_{X_1,...,X_n}(x_1...,x_n|N(T)=n)\frac{(2T+2x_1+1+x_2)x_2}{2},\\
	\mathbb{E}(Q_i)=& \sum_{x_1=1}^{b_1}\sum_{x_2=1}^{b_2}...\sum_{x_n=1}^{b_n}   P_{X_1,...,X_n}(x_1...,x_n|N(T)=n)\frac{(2x_{i-1}+1+x_i)x_i}{2}, \, (3 \leq i \leq n)\\
	\mathbb{E}(Q_{n+1})=&\sum_{x_1=1}^{b_1}\sum_{x_2=1}^{b_2}...\sum_{x_n=1}^{b_n}   P_{X_1,...,X_n}(x_1...,x_n|N(T)=n)\frac{(2x_n+1+x_{n+1})x_{n+1}}{2},
\end{align}where, $x_{n+1}=T-\sum_{i=1}^nx_i,\,b_1=T-n+1,\,b_i=T-(n-i)-\sum_{m=1}^{i-1}x_m, \, i=2,3....n.$\\
\begin{align}
	\mathbb{E}(Q_1)&= \frac{(T+1)[(n+3)T+2]}{(n+1)(n+2)}, \\
	\mathbb{E}(Q_2)&= \frac{(T+1)[(n+4)T+4]}{(n+1)(n+2)}, \\
	\mathbb{E}(Q_3)&=\mathbb{E}(Q_4)=\mathbb{E}(Q_5)=...=\mathbb{E}(Q_{n})=\frac{(T+1)(T+2)}{(n+1)(n+2)} ,  \\
	\mathbb{E}(Q_{n+1})&= \frac{2T(2T+3)-n^2-3n(T+2)}{2(n+1)(n+2)}.
\end{align}
Therefore, the mean area covered by the AoI waveform of the Geo-D system in state $(0, n)$ is
\begin{equation}
	\label{Q0N}
	\mathbb{E}[Q|0,n]=\sum_{i=1}^{n+1}\mathbb{E}[Q_i]=\frac{-n^2+n(8T^2+13T+2)+10T^2+8T-4}{2(n+1)(n+2)}.
\end{equation}

State(k,0): Sensor A sent $ k \,(k \geq 1) $ updates in the previous service period, yet it did not send new updates in the current period. In this case, the instantaneous AoI curve for the Geo-D system is shown in Figure \ref{zhuangtai} \subref{newk0}. When we analyse the state of current,
let $X_i\,(i = 1, 2, ..., n)$ denote the service time of the $i$-th
update in the current period. Within a period $T$, the remaining time after sensor A successfully sent n updates is $X_{n+1}$. As the same way, let $Y_i$ $(i = 1, 2, ..., n)$ denote the service time of the $i$-th
update in the previous period. Within a period $T$, the remaining time after sensor A successfully sent n updates is $Y_{n+1}$.
As shown in Figure \ref{zhuangtai} \subref{newk0}.\par
The mean of PAoI can be calculated as:
\begin{align}
	\label{AK0}
	\mathbb{E}(A|k,0)&= \sum_{y_1=1}^{T-k+1}\sum_{y_2=1}^{T-k+2-y_1}...\sum_{y_k=1}^{T-\sum_{i=1}^{k-1}y_i}   P_{Y_1,...,Y_n}(y_1...,y_n|\widetilde{N}(T)=k)\cdot(T+y_k+y_{k+1})  \nonumber \\
	&=\frac{k(T-1)+3T+1}{k+1} .
\end{align}\par
The average of the area covered by the AoI waveform at the monitor is:
\begin{align}
	\label{QK0}
	\mathbb{E}(Q|k,0)&= \sum_{y_1=1}^{T-k+1}...\sum_{y_k=1}^{T-\sum_{i=1}^{k-1}y_i}   P_{Y_1,...,Y_n}(y_1...,y_n|\widetilde{N}(T)=k)\cdot\frac{T(2y_k+2y_{k+1}+1+T)}{2} \nonumber \\
	&=\frac{T[(k+5)T+4]}{2(k+1)}.
\end{align}\par
For the sake of brevity, we define the upper limit of the i-th summation of $y_i$ as:
\begin{align}
	c_1=T-k+1,c_i=T-(k-i)-\sum_{m=1}^{i-1}y_m, i=2,3....k.
\end{align}

State(k,1): 
Sensor A sent $k\,(k \geq 1)$ updates in the previous service period, and transmitted
one update in the current period, as shown in Figure \ref{zhuangtai} \subref{newk1}. 
Due to the fact that PAOI is calculated in rounds, in this state, the number of updates to PAOI may vary due to simultaneous updates from sensor A and sensor B.

One scenario is that sensors A and B are not updated simultaneously, i.e., $X_1 \neq T$. We use foot tags sub1(2) to distinguish between these two situations. In the current service period, the monitor
received two valid updates, one from sensor A and the other from sensor B. As shown in Figure \ref{zhuangtai} \subref{newk0}.\par
The expressions for each PAoI can be obtained as follows:
\begin{align}
	A_{1,sub1}&=Y_k+Y_{k+1}+X_1,\\
	A_{2,sub1}&=Y_{k+1}+T.
\end{align}\par
Then the expectation can be calculated as:
\begin{align}
	\mathbb{E}(A_{1,sub1})&=\sum_{x_1=1}^{T-1}\sum_{y_1=1}^{c_1}...\sum_{y_k=1}^{c_k}P_{Y_1,...,Y_n}(y_1...,y_n|\widetilde{N}(T)=k)P_{X_1}(x_1|N(T)=1) \nonumber\\
	&\qquad\cdot(y_k+y_{k+1}+x_1) \nonumber \\
	&=\frac{kT-k+5T-3}{2k+2}, \\
	\mathbb{E}(A_{2,sub1})&=\sum_{x_1=1}^{T-1}\sum_{y_1=1}^{c_1}...\sum_{y_k=1}^{c_k}P_{Y_1,...,Y_n}(y_1...,y_n|\widetilde{N}(T)=k)P_{X_1}(x_1|N(T)=1)\cdot(y_{k+1}+T) \nonumber \\
	&=\frac{T-1}{2}+\frac{T-1}{T}(T-\frac{k(T+1)}{T(k+1)}).
\end{align}

Another scenario is that sensors A and B are updated simultaneously, $X_1 = T$. The sensor only received one update from A.
Then the expectation can be calculated as:
\begin{align}
	\mathbb{E}(A_{1,sub2})&=\sum_{x_1=T}^{T}\sum_{y_1=1}^{c_1}...\sum_{y_k=1}^{c_k}P_{Y_1,...,Y_n}(y_1...,y_n|\widetilde{N}(T)=k)P_{X_1}(x_1|N(T)=1) \nonumber\\
	&\qquad\cdot(y_k+y_{k+1}+x_1) \nonumber \\
	&=2-\frac{k(T+1)}{T(k+1)}.
\end{align}

Hence, we can obtain the sum of the average PAoI in this cases as:
\begin{equation}
	\mathbb{E}[A|k,1]=\frac{T-1}{T}\sum_{i=1}^{2}\mathbb{E}(A_{i,sub1})+\frac{1}{T}\mathbb{E}(A_{1,sub2})+=\frac{3T-1[k(T-1)+3T+1]}{2T(k+1)}.
\end{equation}\par
The respective expressions for the waveform coverage area $Q_1$ and $Q_2$  in Figure \ref{zhuangtai} \subref{newk1} can be computed as:
\begin{align}
	\label{QK1}
	\mathbb{E}(Q_1)&= \sum_{x_1=1}^{T}\sum_{y_1=1}^{c_1}...\sum_{y_k=1}^{c_k}   P_{Y_1,...,Y_n}(y_1...,y_n|\widetilde{N}(T)=k)P_{X_1}(x_1|N(T)=1) \nonumber   \\ 
	&\qquad\cdot  [(y_k+y_{k+1}+\frac{x_1}{2})x_1]         \nonumber \\
	&=\frac{(T+1)(2kT+k+14T+13)}{12(k+1)},
\end{align}
\begin{align}
	\mathbb{E}(Q_2)&= \sum_{x_1=1}^{T}\sum_{y_1=1}^{c_1}..\sum_{y_k=1}^{c_k}   P_{Y_1,...,Y_n}(y_1..,y_n|\widetilde{N}(T)=k)P_{X_1}(x_1|N(T)=1) \nonumber   \\ 
	&\qquad\cdot[\frac{(2y_{k+1}+x_1+1+T)(T-x_1)}{2}]        \nonumber \\
	&=\frac{3k^2T-kT^2+k+2T^2-3T-2}{6(k+1)}.
\end{align}\par
Therefore, the average area under the instantaneous AoI curve of the Geo-D system in state (k, 1) is given by
\begin{equation}
	\label{AK1}
	\mathbb{E}[Q|k,1]=\sum_{i=1}^{2}\mathbb{E}(Q_i)=\frac{2k^2T+kT+k+6T^2+7T+3}{4k+4}.
\end{equation}

State(k,n):	Sensor A transmitted $k \,(k \geq 1)$ updates in the previous service period, and
transmits $n \,(n \,\geq 2)$ updates in the current period. As shown in Figure \ref{zhuangtai} \subref{newkn}.\par
The expressions of each PAoI are as follows:
\begin{align}
	A_1&=Y_k+Y_{k+1}+X_1,\\     
	A_2&=Y_{k+1}+X_1+X_2,\\
	A_i&=X_{i-1}+X_i\,(3\leq i \leq n)    .
\end{align}\par
Referring to the same calculation procedure above, it is obtained that:
\begin{align}
	\mathbb{E}(A_1)=& \sum_{x_1=1}^{b_1}..\sum_{x_k=1}^{b_n}\sum_{y_1=1}^{c_1}...\sum_{y_k=1}^{c_k}(y_k+y_{k+1}+x_1)\cdot \nonumber   \\ 
	&\qquad\cdot P_{X_1,..,X_n}(x_1,..,x_n|N(T)=n)\cdot P_{Y_1,..,Y_k}(y_1..,y_k|\widetilde{N}(T)=k)        \nonumber \\
	&=T+\frac{T+1}{n+1}-\frac{(k-1)(T+1)}{k+1}.
\end{align}
Similarly, we have
\begin{align}
	\mathbb{E}(A_2)&=-\frac{k(T+1)}{k+1}+\frac{2(T+1)}{n+1}+T,\\
	\mathbb{E}(A_3)&=\mathbb{E}(A_4)=...=\mathbb{E}(A_{n}) =\frac{2(T+1)}{n+1} .
\end{align}
Therefore, the sum of the average PAoI for the Geo-D system in state (k, n) is:
\begin{equation}
	\mathbb{E}[A|k,n]=\sum_{i=1}^{n}\mathbb{E}(A_i)=\frac{k[(2n-1)T-3]+n(5T+3)+2T}{(k+1)(n+1)}.
\end{equation}
The respective expressions for the waveform coverage area Q are:
\begin{align}
	\label{AKN}
	Q_1&=\frac{(Y_k+Y_{k+1}+1+Y_k+Y_{k+1}+X_1)X_1}{2},\\     
	Q_2&=\frac{(Y_{k+1}+X_1+1+Y_{k+1}+X_1+X_2)X_2}{2},\\
	Q_i&=\frac{(X_{i-1}+1+X_{i-1}+X_{i})X_i}{2},(3\leq i \leq n),\\
	Q_{n+1}&=\frac{(X_n+1+X_n+X_{n+1})X_{n+1}}{2}.
\end{align}
We can get the following expectations:
\begin{align}
	\mathbb{E}[Q_1]&=\frac{(T+1){n(-k+2T+1)+(k+5)T+4}}{(k+1)(n+1)(n+2)},\\
	\mathbb{E}[Q_2]&=\frac{(T+1){k(-n+2T+2)+(n+4)T+4}}{(k+1)(n+1)(n+2)},\\
	\mathbb{E}(Q_3)&=\mathbb{E}(Q_4)=\mathbb{E}(Q_5)=...=\mathbb{E}(Q_{n})=\frac{(T+1)(T+2)}{(n+1)(n+2)},  \\
	\mathbb{E}(Q_{n+1})&= \frac{2T(2T+3)-n^2-3n(T+2)}{2(n+1)(n+2)}.
\end{align}
Hence, we have the average area covered by AoI waveform of the Geo-D system in state (k, n) given as:
\begin{multline}       %双栏准备
	\text{$\mathbb{E}[Q|k,n]=\sum_{i=1}^{n+1}\mathbb{E}(Q_i)=-\frac{1}{2(k+1)(n+1)(n+2)}\cdot$} \\
	\text{$ \left\{k[n^2+n(-4^2-5T+2)-2T^2+8T+12]+\right.                        $    }\\
	\text{$ \left. n^2-n[10T^2+17T+4-2T(7T+8)]           \right\}                           $                         }
\end{multline}

			\begin{table*}[htp]
				\setlength{\tabcolsep}{4mm}
				\renewcommand\arraystretch{1.6}
				\centering
				\caption{The expectations  $\mathbb{E}[A|(k,n)]$, $\mathbb{E}[V|(k,n)]$, and  $\mathbb{E}[Q|(k,n)]$ for different $k$ and $n$.  }
				\begin{tabular}{|c|c|c|c|}
					\hline
					State	& $k=0,\,n=0$ & $k=0 ,\,n=1$&$k=0,\,2\leq n \leq T$\\ 
					\hline
					\rowcolor[gray]{.9}
					$\mathbb{E}[A|(k,n)]$ &$2T$&$2T$   &$T+\frac{2(n-1)(T+1)}{n+1} $\\
					\hline
					$\mathbb{E}[V(k,n)]$ &$1$&$1$   &$n-1 $\\
					\hline
					\rowcolor[gray]{.9}
					$\mathbb{E}[Q|(k,n)]$ &$T(\frac{3T+1}{2})$&$T(\frac{3T+1}{2})$  &$\frac{(T+1)[-2+5T+n(2+4T)]}{(n+1)(n+2)}$  \\
					\hline
					%\hline
					%		\multicolumn{4}{|c|}{\multirow{1}{*}{$k\geq1$}}  \\
					\hline
					State	& $ 1 \leq k \leq T,\,n=0$ & $ 1 \leq  k \leq T,\,n=1$&$1\leq k \leq T,\,2\leq n \leq T$\\
					\hline
					\rowcolor[gray]{.9}
					$\mathbb{E}[A|(k,n)]$& $\frac{k(T-1)+3T+1}{k+1}$ &  $\frac{(3T-1)[3T+1+k(T-1)]}{2T(k+1)}$&$\frac{k[(2n-1)T-3]+n(5T+3)+2T}{(k+1)(n+1)}$\\
					\hline
					$\mathbb{E}[V(k,n)]$ &$1$&$2-\frac{1}{T} $   &$n $\\
					\hline
					\rowcolor[gray]{.9}
					$\mathbb{E}[Q|(k,n)]$ &$\frac{T[3+5T+k(T-1)]}{2(k+1)}$&$\frac{(4+k)T^2-(k-3)T+1}{2(k+1)}$ &$\frac{(T+1)[k(2nT+T-6)+n(5T+3)+7T]}{(k+1)(n+1)(n+2)} $  \\
					\hline
				\end{tabular}	
				\label{tablepaoi} 
			\end{table*}
		\end{proof} 
			\subsection{Average PAoI of the Geo-D System}
		In a service period, the average PAoI of the Geo-D system can be obtained by dividing the sum
		of PAoI by the number of PAoI. The average number of PAoI of the Geo-D system in
		a period can be calculated by
		\begin{equation}
			\mathbb{E}[V]=        \sum_{k=0}^{T}\sum_{n=0}^{T} \Pr\left\lbrace \widetilde{N}(T)=k   \right \rbrace \cdot \Pr\left\lbrace {N}(T)=n   \right \rbrace\cdot  \mathbb{E}[V|k,n].
		\end{equation}\par
		The sum of the average PAoI for the Geo-D system is calculated as
		\begin{equation}
			\label{ea}
			\mathbb{E}[A]=        \sum_{k=0}^{T}\sum_{n=0}^{T} \Pr\left\lbrace \widetilde{N}(T)=k   \right \rbrace \cdot \Pr\left\lbrace {N}(T)=n   \right \rbrace\cdot  \mathbb{E}[A|k,n].
		\end{equation}
			According to \ref{state}, we can obtain Table \ref{tablepaoi}, and substituting the values into \eqref{en} nad \eqref{ea}.   
		\begin{align}
			\mathbb{E}[V]&=q^{2T\!-\!1}\!+\!q^{T\!-\!1}(T\!-\!1)p\!+\!Tp \\
			\mathbb{E}[A]&=\bigg[2T\!+\!q^{2T\!-\!1}\Big[\frac{2}{p}\!-\!(\frac{p}{2}\nonumber 
			\!-\!2)(T\!-\!1)\Big]\!+\!q^{T\!-\!1}\Big(2\!-\!\frac{2}{p}\!-\!\frac{p}{2}\!-\!\frac{Tp}{2}\!+\!T^2p\Big)   \bigg]
		\end{align}\par
		Therefore, the average PAoI of the Geo-D system can be calculated by
		\begin{equation}
			{\rm\Delta_{Geo-D}^{peak}}=\frac{\mathbb{E}[A]}{\mathbb{E}[V]}=\frac{1}{q^{2T\!-\!1}\!+\!q^{T\!-\!1}(T\!-\!1)p\!+\!Tp}\bigg[2T\!+\!q^{2T\!-\!1}\Big[\frac{2}{p}\!-\!(\frac{p}{2}\nonumber 
			\!-\!2)(T\!-\!1)\Big]\!+\!q^{T\!-\!1}\Big(2\!-\!\frac{2}{p}\!-\!\frac{p}{2}\!-\!\frac{Tp}{2}\!+\!T^2p\Big)   \bigg], \\
		\end{equation}

			\subsection{Average AoI of the Geo-D System}
		The average AoI of Geo-D system can be obtained by multiplying the probability of the system in each state with the average AoI in that state.
		Where, the average AoI under each state is the waveform coverage area Q divided by the service period time $T$.
		Then, the average AoI of Geo-D system is obtained by the following formula:
		\begin{align}       %双栏准备
				{ \overline{\rm \Delta}_\text{Geo-D}}=\sum_{k=0}^{T}\sum_{n=0}^{T} \Pr\left\lbrace \widetilde{N}(T)=k   \right \rbrace \cdot \Pr\left\lbrace {N}(T)=n   \right \rbrace\cdot  \frac{\mathbb{E}[Q|k,n]}{T}.
			\label{EQ}
		\end{align}
		According to \ref{state}, we can obtain Table \ref{tablepaoi}, and substituting the values into \eqref{EQ}.

		Therefore, the analytic expression for the average AoI of the Geo-D system is as follows:
		\begin{align}
				{ \overline{\rm \Delta}_\text{Geo-D}}=&\frac{2}{p}\!+\!\frac{q^T}{p}\Big(\!-\!1\!+\!\frac{2}{T}\!-\!\frac{3}{pT}\Big)\!+\!\frac{q^{2T}}{p}\Big(2\!-\!\frac{2}{T}\!+\!\frac{3}{pT}\Big).
		\end{align}

		\subsection{Discussions}\label{Discussions}
		%			Note that in essence, Theorems \ref{Pro1} presents the average PAoI and AoI as functions of the service rates $\mu_{\text{A}}=p$ and  $\mu_{\text{B}}=1/T$.
		One can verify that the result presented in  Theorem \ref{Pro1}  also applies to the boundary cases, i.e., $p=1, T=1, p\to0,$ and $T\to \infty$. Recall that when $p\to 0$ and $T\to \infty$, the Geo-D system degenerates into a ZW/D/1 queue and  a ZW/Geo/1 queue, respectively. Hence, Theorem \ref{Pro1} provides the basis for evaluating the improvement of information freshness brought from the redundancy of each sensor.
		Let us define the reduction ratio of average AoI and PAoI for the Geo-D queue system regarding that of a specific queue $\psi$ with the same service rate as $\eta_{\psi}$ and $\eta_{\psi}^{\text{p}}$, respectively, i.e.,  
		\begin{align}\label{etapsi}
			\eta_{\psi}:=&(\overline\Delta_{\psi}-\overline\Delta_{\text {Geo-D}})/\overline\Delta_{\psi},   \\ 
			\eta_{\psi}^{\text{p}}:=&(\overline\Delta_{\psi}^{\text{p}}-\overline\Delta_{\text {Geo-D}}^{\text{p}})/\overline\Delta_{\psi}^{\text{p}}.\label{etapsip}
		\end{align}
		For fairness, let us consider that $\mu_{\text A}=\mu_{\text {B}}=\mu$. Then, it has $p=\mu, T=\frac{1}{\mu}$, and
		\begin{align}
			{\overline{\Delta}^{\text p}_{\text{Geo-D}}}=& \frac{1}{\mu}(1-\mu)^{\frac{1}{\mu}-1}\Big[(4-\mu)(1-\mu)^{\frac{1}{\mu}+1}- \nonumber \\
			&2(1-\mu)^{\frac{1}{\mu}}(\mu-1)^{1-\frac{2}{\mu}}  +3\mu-\mu^2-1        \Big],  \\
			{\overline\Delta_{\text{Geo-D}}}=&\frac{1}{\mu}\Big[2\!+\!(5\!-\!2\mu)(1\!-\!\mu)^{\frac{2}{\mu}}\!+\! 2(1\!-\!\mu)^{\frac{1}{\mu}}(\mu\!-\!2) \Big].
		\end{align}
		When $\mu_{\text{A}}=p$ and $\mu_{\text{B}}=0$, it implies that $p=\mu, T\!\to \!\infty$, and 
		\begin{align}
			{\overline{\Delta}^{\text p}_{\text{ZW/Geo/1}}}=& \frac{2}{\mu},	\label{PAoIGeo}  \\
			{\overline\Delta_{\text{ZW/Geo/1}}}=&\frac{2}{\mu}\label{AoIGeo}.
		\end{align}
		When $\mu_{\text{A}}=0$ and $\mu_{\text{B}}=\mu$, it implies that $p=0, T=\frac{1}{\mu}$, and 
		\begin{align}
			{\overline{\Delta}^{\text p}_{\text{ZW/D/1}}}=&\frac{2}{\mu}, \label{PAoID}  \\
			{\overline\Delta_{\text{ZW/D/1}}}=&\frac{3}{2\mu}+\frac{1}{2}. \label{AoID}
		\end{align}
		Based on \eqref{etapsi}-\eqref{AoID}, one can know that when $\mu \in (0,1)$
		\begin{align}  
			\eta_{\text{ZW/Geo/1}}=&\big[{\overline\Delta_{\text{ZW/Geo/1}}}-\overline{\Delta}_{\text{Geo-D}}\big]/{\overline\Delta_{\text{ZW/Geo/1}}} \nonumber \\
			=&(2-\mu)(1-\mu)^\frac{1}{\mu}+(\mu-\frac{5}{2})(1-\mu)^\frac{2 }{\mu},   \label{reductionGeo}  \\
			\eta_{\text{ZW/D/1}}=&\big[{\overline\Delta_{\text{ZW/D/1}}}-\overline{\Delta}_{\text{Geo-D}}\big]/{\overline\Delta_{\text{ZW/D/1}}} \nonumber \\
			=&\frac{\!\mu\!-\!1\!+\!4(\!2\!-\!\mu\!)(\!1\!-\!\mu\!)^\frac{1}{\mu}\!+\!2(2\mu\!-\!5)(1\!-\!\mu)^\frac{2}{\mu} }{\mu\!+\!3}  \label{reductionD} .
			%		[{\mu-1 +4(2-\mu)(1-\mu)^\frac{1}{\mu}+2(2\mu-5)(1-\mu)^\frac{2}{\mu} }]/({\mu+3}),
		\end{align}
		In particular, when $\mu \to 0$, we can derive that
		\begin{align}  
			&\lim\limits_{\mu\to0} \eta_{\text{ZW/Geo/1}}=\frac{2}{\mathrm{e}}-\frac{5}{2\mathrm{e}^2}\approx39.74\%,\label{0Geo}   \\
			&\lim\limits_{\mu\to0} \eta_{\text{ZW/D/1}}	=\frac{8}{3\mathrm{e}}-\frac{10}{3\mathrm{e}^2}-\frac{1}{3}\approx19.65\%.\label{0D} 
		\end{align}
		%			Note that when $\mu \to 0$, the average AoI of a single-queue system would increase dramatically. \eqref{0Geo} indicates that a significant 39.74$\%$ reduction of the average AoI can be achieved from adding a ZW/D/1 queue for a ZW/Geo/1 queue. \eqref{0D} indicates that a 19.65$\%$ reduction can be achieved from adding a ZW/Geo/1 queue for a ZW/D/1 queue. While $\mu$ can be other values in $(0,1)$, \eqref{0Geo} and \eqref{0D} reveal that the improvement of average AoI brought from a ZW/D/1 queue and a ZW/Geo/1 queue can be up to  39.74$\%$ and 19.65$\%$, respectively.
		Note that when $\mu \to 0$, the average AoI of a single-queue system would increase dramatically. In \eqref{0Geo} (resp. \eqref{0D}), it indicates that a significant 39.74$\%$ (resp.19.65\%) reduction of the average AoI can be achieved from adding a ZW/D/1 queue (resp. ZW/Geo/1 queue ) for a ZW/Geo/1 queue (resp. ZW/D/1 queue). 
		%						\eqref{0D} indicates that a 19.65$\%$ reduction can be achieved from adding a ZW/Geo/1 queue for a ZW/D/1 queue. While $\mu$ can be other values in $(0,1)$, \eqref{0Geo} and \eqref{0D} reveal that the improvement of average AoI brought from a ZW/D/1 queue and a ZW/Geo/1 queue can be up to  39.74$\%$ and 19.65$\%$, respectively. 
		
		Similarly, one can also know that 
		\begin{align}  
			\eta^{\text{p}}_{\text{ZW/Geo/1}}=& \,\eta^{\text{p}}_{\text{ZW/D/1}} \nonumber \\
			=&\frac{{(\mu\!-\!2)(1\!-\!\mu)^\frac{2}{\mu}\!+\!(1\!-\!\mu)^{\frac{1}{\mu}-1}}(3\mu\!-\!1\!-\!\mu^2)}{{2(1\!-\!\mu)^\frac{2}{\mu}\!+\!(1\!-\!\mu)^{\frac{1}{\mu}\!-\!1}\!+\!1}},
		\end{align} where $\mu \in(0,1)$. In particular, 
		\begin{align}  
			\lim\limits_{\mu \to 0}\eta^{\text{p}}_{\text{ZW/Geo/1}}=\lim\limits_{\mu \to 0}\eta^{\text{p}}_{\text{ZW/D/1}}\!=\!\frac{3\mathrm{e}\!-\!2}{2(\mathrm{e}^2\!+\!\mathrm{e}\!+\!1)}\!\approx\!27.71\%,
		\end{align}
		which shows that the reduction of average PAoI can be significantly up to    $27.71\%$ by adding  a ZW/Geo/1 queue to a ZW/D/1 queue and vise versa.
		
		\section{Connection Between Discrete-time and Continuous-time Systems}\label{VI}
		Note that the geometric distribution and Bernoulli process are widely regarded as discrete counterparts  of exponential distribution and Poisson process, respectively. In this section, we show that the continuous-time system is just a limit case of the  discrete-time system in terms of AoI, and the result of the discrete-time system is more general than that of the continuous-time system. For clarity, let us specify $\overline\Delta_{\psi}$ and $\overline\Delta_{\psi}^{\text{p}}$ of a queue as functions of $\mu_\text{A}$ and $\mu_\text{B}$, i.e., $\overline\Delta_{\psi}(\mu_{\text{A}},\mu_{\text{B}})$ and $\overline\Delta_{\psi}^{\text{p}}(\mu_{\text{A}},\mu_{\text{B}})$. 
		%		Note that the geometric and exponential distributions have memory-less properties, we  discuss the  relation between discrete-time and continuous-time systems in this section.

		%		In fact, by replacing  the ZW/Geo/1 queue in the continuous-time system with a ZW/M/1 queue, one can get a parallel queue system named M-D system studied in \cite{mmmd}.	
		
		By replacing the ZW/Geo/1 queue with a ZW/M/1 queue and relaxing the integer constraint of $T$ to a positive real number, the considered system changes to a continuous-time dual-queue system called M-D system whose average AoI and PAoI were studied in \cite{mmmd}. An M-M system where both queues are ZW/M/1 queues was also investigated in \cite{mmmd}. Note that the discrete-time counterpart of the M-M system can be obtained from the Geo-D system by replacing the ZW/D/1 queue with another ZW/Geo/1 queue.  
		
		\subsection{Connection between the Geo-D System and the M-D System}
		To study the connection between the considered system and its discrete counterpart, let us assume the time slot length is $1/\delta$ second, where $\delta$ represents the degree to which one second is divided.
		%			 Accordingly, we multiply the AoI results presented in Theorem 1 with factor $1/\delta$ in this section.
		Then, we have the following result.
		
		% \begin{proposition}	\label{T3}Consider an M-D system with service rate of sensor A $\lambda$ and that of sensor B $1/T_{\text{M}}$. The average AoI and average PAoI of a Geo-D system and those of an M-D system satisfy that %has the following conversion relationship with the expression of average AoI and average PAoI in Geo-D system, where $\frac{1}{\delta}$ represents the length of the time slot, and $\delta$ represents the scale of time slot reduction.
			\begin{Theorem}			
			\label{T3}Consider an M-D system with service rate of  $\mu_{\text{A}}=\lambda$ and $\mu_{\text{B}}=1/T_{\text M}$. It holds that
				\begin{align}\label{AoIMD}
					\lim\limits_{\delta \to \infty}\frac{1}{\delta}{\overline{\Delta}_{\text {Geo-D}}}\Big(\frac{\lambda}{\delta},\frac{1}{\delta T_{\text{M}}}\Big) &=\overline{\Delta}_{\text{M-D}}\Big(\lambda,\frac{1}{ T_{\text{M}}}\Big), \\
					\label{pAoIMD}	\lim\limits_{\delta \to \infty}\frac{1}{\delta}{\overline{\Delta}^{\text{p}}_{\text {Geo-D}}}\Big(\frac{\lambda}{\delta},\frac{1}{\delta T_{\text{M}}}\Big) &=\overline{\Delta}^{\text{p}}_{\text{M-D}}\Big(\lambda,\frac{1}{ T_{\text{M}}}\Big).
				\end{align}
				\end{Theorem}	
			\begin{proof}
				%					According to the previous analysis, we have the following equations
				Based on \eqref{exaoi}, it has
				\begin{align}
					&\lim\limits_{\delta \to \infty}\frac{1}{\delta}{\overline{\Delta}_{\text {Geo-D}}}\Big(\frac{\lambda}{\delta},\frac{1}{\delta T_{\text{M}}}\Big) \nonumber \\
					=&\lim\limits_{\delta \to \infty}\frac{1}{\delta}\bigg[\frac{2}{\frac{\lambda}{\delta}}+\frac{(1-\frac{\lambda}{\delta})^{\delta  T_{\text{M}}}}{\frac{\lambda}{\delta}}  \big(-1+\frac{2}{\delta  T_{\text{M}}}-\frac{3}{\frac{\lambda}{\delta}\cdot\delta  T_{\text{M}}}\big) \nonumber \\
					&+	\frac{(1-\frac{\lambda}{\delta})^{2\delta T_{\text{M}}}}{\frac{\lambda}{\delta}} \big(2-\frac{2}{\delta  T_{\text{M}}}+\frac{3}{\frac{\lambda}{\delta}\cdot\delta T_{\text{M}}}\big)\bigg]     \nonumber    \\
					=& \frac{ 3+2T_{\text{M}}\lambda+\mathrm{e}^{T_{\text{M}}\lambda}[-3+(-1+2\mathrm{e}^{T_{\text{M}}\lambda})T_{\text{M}}\lambda]}{T_{\text{M}}\lambda^2\mathrm{e}^{2T_{\text{M}}\lambda}}      \nonumber \\
					\overset{\text{(a)}}{=}&{\overline{\Delta}_{\text {M-D}}}\Big({\lambda},\frac{1}{T_{\text{M}}}\Big)
				\end{align}
				where (a) follows from  Eq. (4) of \cite{mmmd}. Similarly, 
				\begin{align}
					&\lim\limits_{\delta \to \infty}\frac{1}{\delta}{\overline{\Delta}^{\text{p}}_{\text {Geo-D}}}\Big(\frac{\lambda}{\delta},\frac{1}{\delta T_{\text{M}}}\Big) \nonumber \\
					=&\lim\limits_{\delta \to \infty}\frac{1}{({1-\frac{\lambda}{\delta}})^{2{\delta T_{\text{M}}}-1}+({1-\frac{\lambda}{\delta}})^{{\delta T_{\text{M}}}-1}({\delta T_{\text{M}}}-1){\frac{\lambda}{\delta}}+{\lambda T_{\text{M}}}}    \nonumber \\
					&\bigg[2{\delta T_{\text{M}}}+ ({1-\frac{\lambda} {\delta}})^{2{\delta T_{\text{M}}}-1}\Big({\frac{2\lambda}{\delta}-(\frac{\delta}{2\lambda}-2)(\delta T_{\text{M}}-1)}\Big)\nonumber \\
					&+({1-\frac{\lambda}{\delta}})^{{\delta T_{\text{M}}}-1}\Big(2-\frac{2}{{\frac{\lambda}{\delta}}}-{\frac{\lambda}{2\delta}}-\frac{\delta T_{\text{M}}\lambda}{2\delta}+{\delta^2 T^2_{\text{M}}}{\frac{\lambda}{\delta}} \Big)                             \bigg]\cdot \frac{1}{\delta}   \nonumber \\
					=&\frac{2+2\lambda T_{\text{M}}+\mathrm{e}^{\lambda T_{\text{M}}}[-2+\lambda T_{\text{M}}(2\mathrm{e}^{\lambda T_{\text{M}}}+\lambda T_{\text{M}})]}{\lambda[1+\mathrm{e}^{\lambda T_{\text{M}}}(1+\mathrm{e}^{\lambda T_{\text{M}}})\lambda T_{\text{M}}]} \nonumber \\
					\overset{\text{(a)}}{=}&{\overline{\Delta}^{\text{p}}_{\text {M-D}}}\Big({\lambda},\frac{1}{T_{\text{M}}}\Big),
				\end{align} 
				% \begin{align}
					% 	&\lim\limits_{\delta \to \infty}\frac{1}{\delta}{\overline{\Delta}^{\text{p}}_{\text {Geo-D}}}(\frac{\mu_{\text A}}{\delta}, \frac{\mu_{\text B}}{\delta}) \nonumber \\
					% 	&=\lim\limits_{\delta \to \infty}\frac{1}{({1-\frac{\lambda}{\delta}})^{2{\delta T}-1}+({1-\frac{\lambda}{\delta}})^{{\delta T}-1}({\delta T}-1){\frac{\lambda}{\delta}}+{\lambda T}}    \nonumber \\
					% 	&\bigg[2{\delta T}+ ({1-\frac{\lambda} {\delta}})^{2{\delta T}-1}\Big({\frac{2\lambda}{\delta}-(\frac{\delta}{2\lambda}-2)(\delta T-1)}\Big)\nonumber \\
					% 	&+({1-\frac{\lambda}{\delta}})^{{\delta T}-1}\Big(2-\frac{2}{{\frac{\lambda}{\delta}}}-{\frac{\lambda}{2\delta}}-\frac{\delta T\lambda}{2\delta}+{\delta T}^2{\frac{\lambda}{\delta}} \Big)                             \bigg]\cdot \frac{1}{\delta}   \nonumber \\
					% 	&=\frac{2+2\lambda T+\mathrm{e}^{\lambda T}[-2+\lambda T(2\mathrm{e}^{\lambda T}+\lambda T)]}{\lambda[1+\mathrm{e}^{\lambda T}(1+\mathrm{e}^{\lambda T})\lambda T]} \nonumber \\
					% 	&\overset{\text{(a)}}{=}{\overline{\Delta}^{\text{p}}_{\text {M-D}}}(\mu_{\text A},\mu_{\text B}),
					% \end{align} 
				where (a) follows from Eq. (3) of \cite{mmmd}.	%This completes the proof of Theorem \ref{T3}.      
			\end{proof}
			
			Note that the factor $1/\delta$ at the left-hand side of \eqref{AoIMD} and that of \eqref{pAoIMD} transform the AoI values in a time slot to second.
			% $T_{\text{M}}$ is just the service time of the ZW/D/1 queue of the M-D system. 
			Theorem \ref{T3} reveals that  the average AoI and PAoI of a dual-queue M-D system is just a limit case of the counterpart of a dual-queue Geo-D system. More specifically, the average AoI and PAoI of a dual-queue M-D system can be approached by  a series of dual-queue Geo-D systems with shrinking time slot length $1/\delta$. To understand this, let us rigorously show how a geometric distribution  evolves into an exponential distribution when the time slot length tends to be infinitesimal.
			
			Specifically, for $l \in \mathbb{N^+}$, suppose that $U_l$ follows the geometric distribution on $\mathbb{N^+}$ with success parameter $p_l=\frac{r}{l}$, where $r>0$. Then the distribution of $M_l:=U_l/l$ converges to the exponential distribution with parameter $r$ as $l \to \infty$.
			This is because
			%				\begin{prop}\label{GeoM}
				%					\textit{For $l \in \mathbb{N^+}$, suppose that $U_l$ follows the geometric distribution on $\mathbb{N^+}$ with success parameter $p_l=\frac{r}{l}$, where $r>0$. Then the distribution of $M_l:=U_l/l$ converges to the exponential distribution with parameter $r$ as $l \to \infty$. }
				%				\end{prop}
			%\begin{proof}
			%Let $M_n=\frac{U_n}{n}$, depending on the definition, we can get
			\begin{align}
				&\lim\limits_{l \to \infty}	\mathbb{P}(M_l\leq m)=\lim\limits_{l \to \infty}\mathbb{P}({U_l} \leq lm) \nonumber \\
				=&\lim\limits_{l \to \infty}1-\bigg(1-\frac{r}{l}\bigg)^{l m} \nonumber \\
				=&1-\mathrm{e}^{-rm}.\label{GtoM}
			\end{align}
			%This completes the proof. 
			%\end{proof}
			Based on  \eqref{GtoM}, it reveals that by appropriately selecting the success parameter, a sequence of geometrically distributed random variables converge to an exponentially distributed random variable. This implies that a sequence of ZW/Geo/1 queues with service rate $\mu_{\text{A}}=\lambda/\delta$ would converge to a ZW/M/1 queue with service rate $\lambda$. Note that $\delta$ not only plays the role of indexing the sequence but also controls the time slot length $1/\delta$. The key to selecting the sequence of ZW/Geo/1 queues is to guarantee that the service rate in a slot is the same as that of the ZW/M/1 queue in the second. One can see that the geometrically-distributed random variables $\{X_i\}$ and $\{Y_j\}$ for different states also converge to the corresponding exponentially-distributed random variables when $\delta\to\infty$. In contrast, for the ZW/D/1 queue, the service period $T$ in a slot of the Geo-D system should be considered as $T_{\text{M}}/(1/\delta)$ since $T_{\text{M}}$ is in second and should be normalized based on the time slot length.	Finally, the convergence presented in Theorem \ref{T3} holds naturally.

			\subsection{Connection between the Geo-Geo System and the M-M System}
			Intuitively, the relation between the Geo-D and the M-D systems in terms of AoI can be extended to that between the Geo-Geo and the M-M systems. 
			
			First, let us derive the average AoI of the Geo-Geo system. 
			\begin{lemma}
				The average AoI of the Geo-Geo system is 
				\begin{multline}\label{eqgeogeo}
					\overline{\Delta}_{\text{Geo-Geo}}(\mu_{\text{A}},\mu_{\text{B}})= \frac{2(\mu_{\text A}^2+\mu_{\text B}^2+3\mu_{\text A}\mu_{\text B}-3\mu_{\text A}^2\mu_{\text B}-3\mu_{\text A}\mu_{\text B}^2+2\mu_{\text A}^2\mu_{\text B}^2)}{(\mu_{\text A}+\mu_{\text B}-\mu_{\text A}\mu_{\text B})^3}.  
				\end{multline} 
			\end{lemma}
			
			\begin{proof} 
				Instead of computing the average AoI of the Geo-Geo system via graphical analysis method, we show it can be derived based on the average age of the stalest information (AoSI). The AoSI is defined as 
				the worse-one instantaneous  AoI of the two parallel queues \cite{geogeo}. %, which is adopted for evaluating the information freshness of the dual-queue system. 
				Note that the AoI and the AoSI of a dual-queue system are essentially the minimum and the maximum of the AoIs of the two ZW/Geo/1 queues. Hence, the instantaneous AoI of the Geo-Geo system ${\Delta_{\text{Geo-Geo}}}$, that of the two ZW/Geo/1 queues, i.e.,  ${\Delta_{\text{A,ZW/Geo/1}}}$ and ${\Delta_{\text{B,ZW/Geo/1}}}$, and the instantaneous AoSI of the Geo-Geo system ${\Delta_{\text{S,Geo-Geo}}}$, satisfy that
				\begin{align}
					\nonumber{\Delta_{\text{Geo-Geo}}}+{\Delta_{\text{S,Geo-Geo}}}=&\min\{{\Delta_{\text{A,ZW/Geo/1}}},{\Delta_{\text{B,ZW/Geo/1}}}\}+\\
					\nonumber&\max\{{\Delta_{\text{A,ZW/Geo/1}}},{\Delta_{\text{B,ZW/Geo/1}}}\}\\
					=&{\Delta_{\text{A,ZW/Geo/1}}}+{\Delta_{\text{B,ZW/Geo/1}}},
				\end{align} 
				which naturally implies that 
				\begin{align}
					&{\overline{\Delta}_{\text{Geo-Geo}}}(\mu_{\text A},\mu_{\text B}) \nonumber \\
					&={\overline{\Delta}_{\text{ZW/Geo/1}}}(\mu_{\text A})\!+\!{\overline{\Delta}_{\text{ZW/Geo/1}}}(\mu_{\text B})\!-\!{\overline{\Delta}_{\text{S,Geo-Geo}}}(\mu_{\text A},\mu_{\text B}). \label{togeogeo} 
				\end{align}

				Specifically,  the AoSI of the Geo-Geo system in \cite{geogeo}   is 
				%					\begin{align}
					%						&{\overline{\Delta}_{\text{S,Geo-Geo}}}(\mu_{\text{A}},\mu_{\text{B}})   \nonumber \\
					%						&=\frac{2}{\mu_{\text A}}+\frac{2}{\mu_{\text B}}-\frac{2}{(\mu_{\text A}+\mu_{\text B}-\mu_{\text A}\mu_{\text B})^3}\big(\mu_{\text A}^2+\mu_{\text B}^2+3\mu_{\text A}\mu_{\text B}- \nonumber \\
					%						&3\mu_{\text A}^2\mu_{\text B}-3\mu_{\text A}\mu_{\text B}^2+2\mu_{\text A}^2\mu_{\text B}^2\big).               
					%					\end{align} 
				\begin{align}
					&{\overline{\Delta}_{\text{S,Geo-Geo}}}(\mu_{\text{A}},\mu_{\text{B}})=\frac{2}{\mu_{\text A}}+\frac{2}{\mu_{\text B}}-\frac{2}{(\mu_{\text A}+\mu_{\text B}-\mu_{\text A}\mu_{\text B})^3} \times   \nonumber \\
					&\big(\mu_{\text A}^2+\mu_{\text B}^2+3\mu_{\text A}\mu_{\text B}- 3\mu_{\text A}^2\mu_{\text B}-3\mu_{\text A}\mu_{\text B}^2+2\mu_{\text A}^2\mu_{\text B}^2\big).               
				\end{align} 			
				while $\overline{\Delta}_{\text{A,ZW/Geo/1}}=2/\mu_{\text{A}}$ and $\overline{\Delta}_{\text{B,ZW/Geo/1}}=2/\mu_{\text{B}}$ can be readily obtained from \eqref{AoIGeo}. Thus, based on  \eqref{togeogeo}, one can see that ${\overline{\Delta}_{\text{Geo-Geo}}}(\mu_{\text A},\mu_{\text B})$ can be computed as \eqref{eqgeogeo}.
				%				\begin{align}
					%					&{\overline{\Delta}_{\text{Geo-Geo}}}(\mu_{\text A},\mu_{\text B}) \nonumber \\
					%					%						&={\overline{\Delta}_{\text{ZW/Geo/1}}}(\mu_{\text A})\!+\!{\overline{\Delta}_{\text{ZW/Geo/1}}}(\mu_{\text B})\!-\!{\overline{\Delta}_{\text{S,Geo-Geo}}}(\mu_{\text A},\mu_{\text B}) \nonumber\\
					%					&=\frac{2(\mu_{\text A}^2\!+\!\mu_{\text B}^2\!+\!3\mu_{\text A}\mu_{\text B}\!-\!3\mu_{\text A}^2\mu_{\text B}\!-\!3\mu_{\text A}\mu_{\text B}^2\!+\!2\mu_{\text A}^2\mu_{\text B}^2)}{(\mu_{\text A}\!+\!\mu_{\text B}\!-\!\mu_{\text A}\mu_{\text B})^3}.
					%					% ~~\hfill{\square} 
					%			\end{align}
			\end{proof}
			Then, the following relation between the average AoI of the Geo-Geo system and that of the M-M system can be found. 
			\begin{Theorem}
			    \label{T4}
				Consider an M-M system with service rate $\mu_{\text{A}}$ and $\mu_{\text{B}}$. It has that
				\begin{align}\label{AoIMM}
					\lim\limits_{\delta \to \infty}\frac{1}{\delta}{\overline{\Delta}_{\text {Geo-Geo}}}\Big(\frac{\mu_{\text{A}}}{\delta},\frac{\mu_{\text{B}}}{\delta}\Big) &=\overline{\Delta}_{\text{M-M}}(\mu_{\text{A}},\mu_{\text{B}}).
				\end{align}
			\end{Theorem}	
			\begin{proof}
				%We find that the  $\mathbb{E}[{\Delta_{\text{Geo-Geo}}}](p_1,p_2)$   and expression of average AoI of  the M-M system in \cite{mmmd} have the  connection as same as that between M-D and Geo-D systems as follows
				\begin{align}
					&\lim\limits_{\delta \to \infty}\frac{1}{\delta}{\overline{\Delta}_{\text {Geo-Geo}}}\Big(\frac{\mu_{\text{A}}}{\delta},\frac{\mu_{\text{B}}}{\delta}\Big) \nonumber \\
					%	&=    \lim\limits_{\delta \to \infty} \frac{1}{\delta}\frac{2(\mu_{\text A}^2+\mu_{\text B}^2+3\mu_{\text A}\mu_{\text B}-3\mu_{\text A}^2\mu_{\text B}-3\mu_{\text A}\mu_{\text B}^2+2\mu_{\text A}^2\mu_{\text B}^2)}{(\mu_{\text A}+\mu_{\text B}-\mu_{\text A}\mu_{\text B})^3}                      \nonumber \\
					= &   \lim\limits_{\delta \to \infty} \frac{1}{\delta}\cdot\frac{2(\frac{\mu_{\text A}^2}{\delta^2}+\frac{\mu_{\text B}^2}{\delta^2}+\frac{3\mu_{\text A}\mu_{\text B}}{\delta^2}-\frac{3\mu_{\text A}^2\mu_{\text B}+3\mu_{\text A}\mu_{\text B}^2}{\delta^3}+\frac{2\mu_{\text A}^2\mu_{\text B}^2}{\delta^4})}{(\frac{\mu_{\text A}}{\delta}+\frac{\mu_{\text B}}{\delta}-\frac{\mu_{\text A}}{\delta}\frac{\mu_{\text B}}{\delta})^3}                      \nonumber \\
					=&   \frac{2(\mu_{\text A}^2+3\mu_{\text A}\mu_{\text B}+\mu_{\text B}^2)}{(\mu_{\text A}+\mu_{\text B})^3}                  \nonumber \\
					\overset{\text{(a)}}{=}&\overline{\Delta}_{\text{M-M}}(\mu_{\text{A}},\mu_{\text{B}}), \label{geotomm}
				\end{align}
				where (a) follows from  Eq. (2) of \cite{mmmd}.
			\end{proof}
			
			%		The detailed computational process is omitted here, please refer to \cite{mmmd} and \cite{geogeo}. 

			\subsection{Discussions}
			Theorems \ref{T3} and \ref{T4} indicate that the continuous-time system is just a limit case of the counterpart discrete-time system in terms of AoI. This manifests that the AoI result of the discrete-time system is more general. Comparing the expressions of the  discrete-time system with those of the continuous-time system, it can be found that  more terms contribute to the average AoI and PAoI of the discrete-time system. These terms become higher-order infinitesimals as the time slot length approaches zero. Hence, studying the AoI of the discrete-time system can provide more details not seen in    the results of the continuous-time system. Note that when evaluating the effect of a redundant queue, a 39.74\% is achieved by adding a ZW/D/1 queue for a ZW/M/1 queue\cite{mmmd}, regardless of the service rate $\mu$. In contrast, for the discrete-time system, the 39.74\% reduction ratio in terms of average AoI is found (cf. \eqref{reductionGeo}) only when $\mu$ goes to $0$. In fact, the reduction ratio is strongly related to the service rate $\mu$ for the discrete-time system. This is  different from the continuous-time system. 

			\section{ Numerical Results}\label{IV}
			In this section,  numerical results are provided to evaluate the performance of the dual-queue Geo-D system. 
				When we equalize the service rate of sensor A with that of sensor B, i.e., $T = 1/p$, we can see the gain in AoI of the dual queue compared to the single queue.\par
			For two single queues, we can see that the change trend of gain is consistent. As the service rate increases, the gain decreases. When the service rate reaches 1, the AoI of all queues reaches the theoretical minimum. Therefore, when the service rate is 1, there is no longer any  gain. This is in line with what is intuitively expected.\par
			It should be noted that, in practice, for a dual queue system, the parameters of fixed service duration and geometric distribution should be equivalent. That's exactly why we make $\mu_A=\mu_B$, to see the AoI gain of the system. There should be no big gap. When the difference between the two service rates is too large, the two-queue system will degenerate into a single-queue system. According to \eqref{AoI}, Table \ref{trend} is obtained, which shows this degradation trend. The expression for the degraded queue is exactly the same as the expression of AoI in.\par
			\begin{table*}[htp]
				\centering
				\caption{Numerical results of AOI under different parameters  } 
				\begin{tabular}{|c|c|c|c|}
					\hline
					& $T\to1$ & $T$&$T\to\infty$\\
					\hline
					
					$p\to1$&$\rm\Delta_{Geo-D}=2$&$2$&$2$ \\
					\hline
					
					$p$&$\rm\Delta_{Geo-D}=2$  & $\rm\Delta_{Geo-D}(p,T)$&$\frac{2}{p}$ \\
					\hline
					
					$p\to0$&$\rm\Delta_{Geo-D}=2$&$\frac{3T+1}{2}$&$\infty$ \\
					\hline
				\end{tabular}
				\label{trend}
			\end{table*}\par
			It is important to note that the range of values of $T$. The equation is based on the premise that $T$ is greater than 2 and is an integer. The various equations in the calculation process are also based on this premise.

			Fig. \ref{sim} depicts the average AoI and PAoI (cf. \eqref{expaoi} and \eqref{exaoi}) of the Geo-D system for different $\mu_{\text {B}}$ when $\mu_{\text A}=p$ ranges from 0 to 1. The simulated results are also presented in Fig. \ref{sim} with the same parameters. Specifically, we simulate the service process of the Geo-D system for 5000 service   periods of sensor B  as one round and repeat 10 rounds using different random seeds  of service time of sensor A.
   Then, the mean value of the 10 rounds of the 5000-period average AoI and average PAoI are taken as the simulation value. It can be seen from Fig. \ref{sim} that the simulation results match well with the theoretical ones, which validates the effectiveness of our analysis. It is further demonstrated that when the difference in the service rate of the two servers is significant, the effect of the weak server on the AoI of the Geo-D system becomes small. The AoI of the dual-queue system is dominated by the strong server. This is consistent with the analysis  of the boundary cases.
			
			%		\begin{figure}[htp] 
				%			%			\subfloat[Theoretical and simulation results of average AoI  \label{aoisim}]{\includegraphics[width=0.9\textwidth]{aoisim.eps}} 
				%			\subfloat[Theoretical and simulation results of average AoI  \label{aoisim}]{\includegraphics[width=0.9\textwidth]{matfig1.eps}} 
				%			\hfill 	
				%			%			\subfloat[Theoretical and simulation results of average PAoI \label{paoisim}]{\includegraphics[width=0.9\textwidth]{paoisim.eps}}  
				%			\subfloat[Theoretical and simulation results of average PAoI \label{paoisim}]{\includegraphics[width=0.9\textwidth]{matfig2.eps}}  
				%			\caption{Theoretical and simulation results of average AoI and average PAoI for Geo-D system under different parameters. }
				%			\label{sim}
				%		\end{figure}

			\begin{figure}[htp] 
				%			\subfloat[Theoretical and simulation results of average AoI  \label{aoisim}]{\includegraphics[width=0.9\textwidth]{aoisim.eps}} 
				\subfloat{\includegraphics[width=0.9\textwidth]{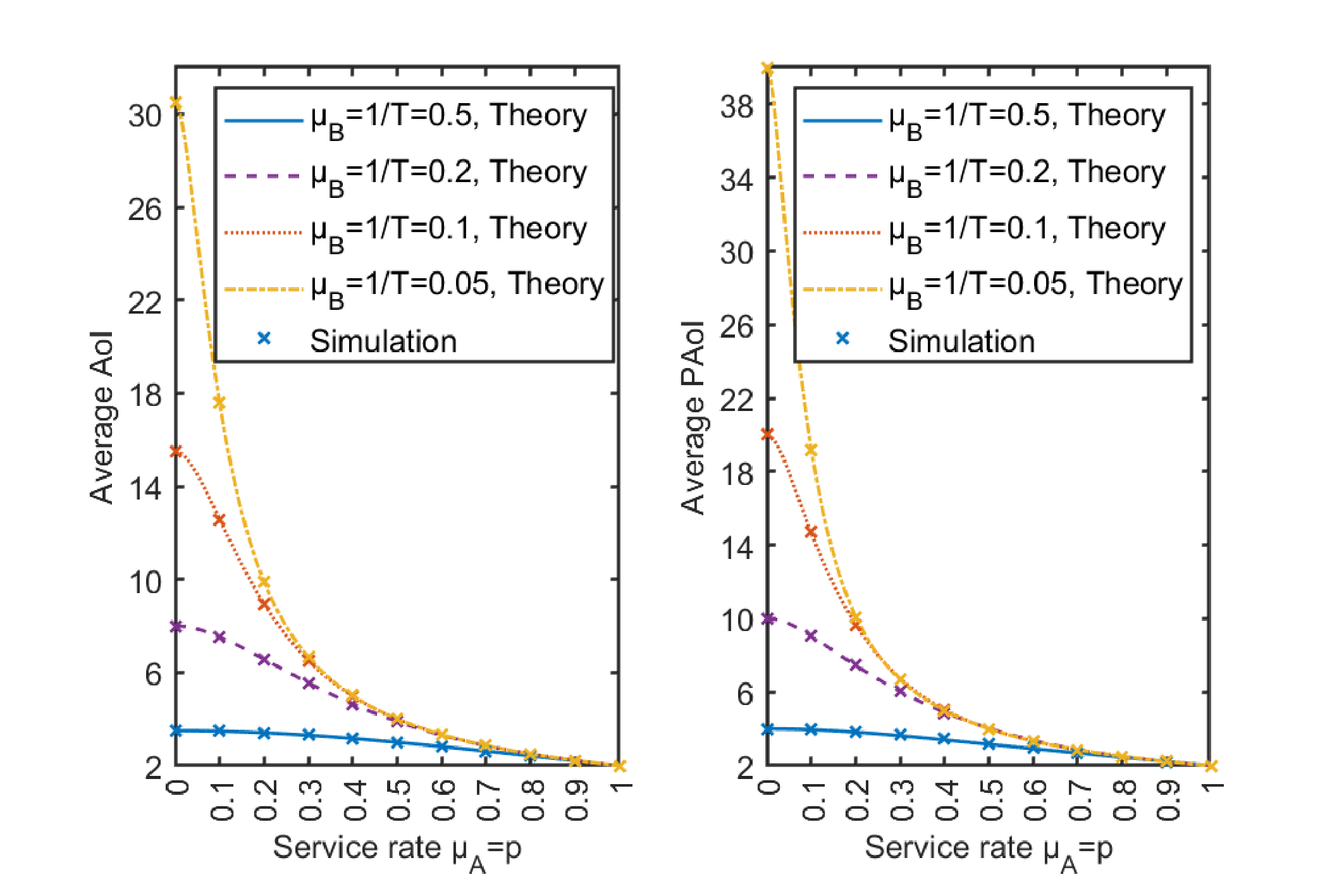}} 
				%			\hfill 	
				%			\subfloat[Theoretical and simulation results of average PAoI \label{paoisim}]{\includegraphics[width=0.9\textwidth]{paoisim.eps}}  
				%			\subfloat[Theoretical and simulation results of average PAoI \label{paoisim}]{\includegraphics[width=0.9\textwidth]{matfig2.eps}}  
				\caption{Theoretical and simulation results of the average AoI and PAoI of Geo-D system under different parameters. }
				\label{sim}
			\end{figure}

			%		 $\overline{\Delta}_{\rm D-D}=b(b-1)/\mu+3/(2\mu)+1/2$, where $0\leq b \leq 1$, represents the offset ratio of service start time between two ZW/D/1 queues. In the following analysis, we set the offset of the difference in service start time between the queues of the D-D system to the optimal, $b=0.5$, to analyze the performance difference between the D-D system and other queuing systems.

			Fig. \ref{aoiu} presents the average AoI of the Geo-D dual-queue system versus the service rate when $\mu_{\text{A}}=\mu_{\text{B}}$. For comparison, the average AoI of the ZW/Geo/1 queue (cf. \eqref{AoIGeo}), that of the ZW/D/1 queue (cf. \eqref{AoID}), and the corresponding reduction ratio \eqref{reductionGeo} and \eqref{reductionD} are also shown. Besides, the average AoI of the Geo-Geo system (cf. \eqref{eqgeogeo}) and that of the D-D system are also given. In particular, in the D-D system, both queues are ZW/D/1 queues, and the average AoI is derived as  $\overline{\Delta}_{\text{D-D}}=\mu+3/(2\mu)-1/2$, where the service start time between two ZW/D/1 queues is considered to be  one slot. It can be seen from Fig. \ref{aoiu} that there is a significant reduction of the AoI of the dual-queue systems compared with that of the two single-queue systems. This is because  introducing  the redundant sensor provides 
the monitor with more frequent valid status updates. More specifically, regarding the Geo-D system, the reduction ratio of the average AoI keeps increasing when the service rate diminishes. It can be understood that when the service rate is low, the ratio of valid status updates to all status updates will increase due to the long update arrival interval of each queue and the asynchronism of the two queues. The reduction ratio resulting from the ZW/D/1 queue is roughly twice  that from the ZW/Geo/1 queue when $\mu_{\text{A}}=\mu_{\text{B}}\leq0.3$. This implies that a low variance of service time  is  valuable for avoiding the deterioration of the average AoI for a single-queue system. Interestingly, the determinacy of the service process is not always good for AoI reduction for dual-queue systems. As  seen from Fig. \ref{aoiu}, both the Geo-Geo and Geo-D systems outperform the D-D system regarding average AoI for a wide range of service rates. This implies that some randomness in the arrival of the updates would  contribute to the AoI reduction.    

			\begin{figure}[htp] 
				\hfill 	
				%			\subfloat{\includegraphics[width=0.9\textwidth]{aoiu.eps}} 
				\subfloat{\includegraphics[width=0.9\textwidth]{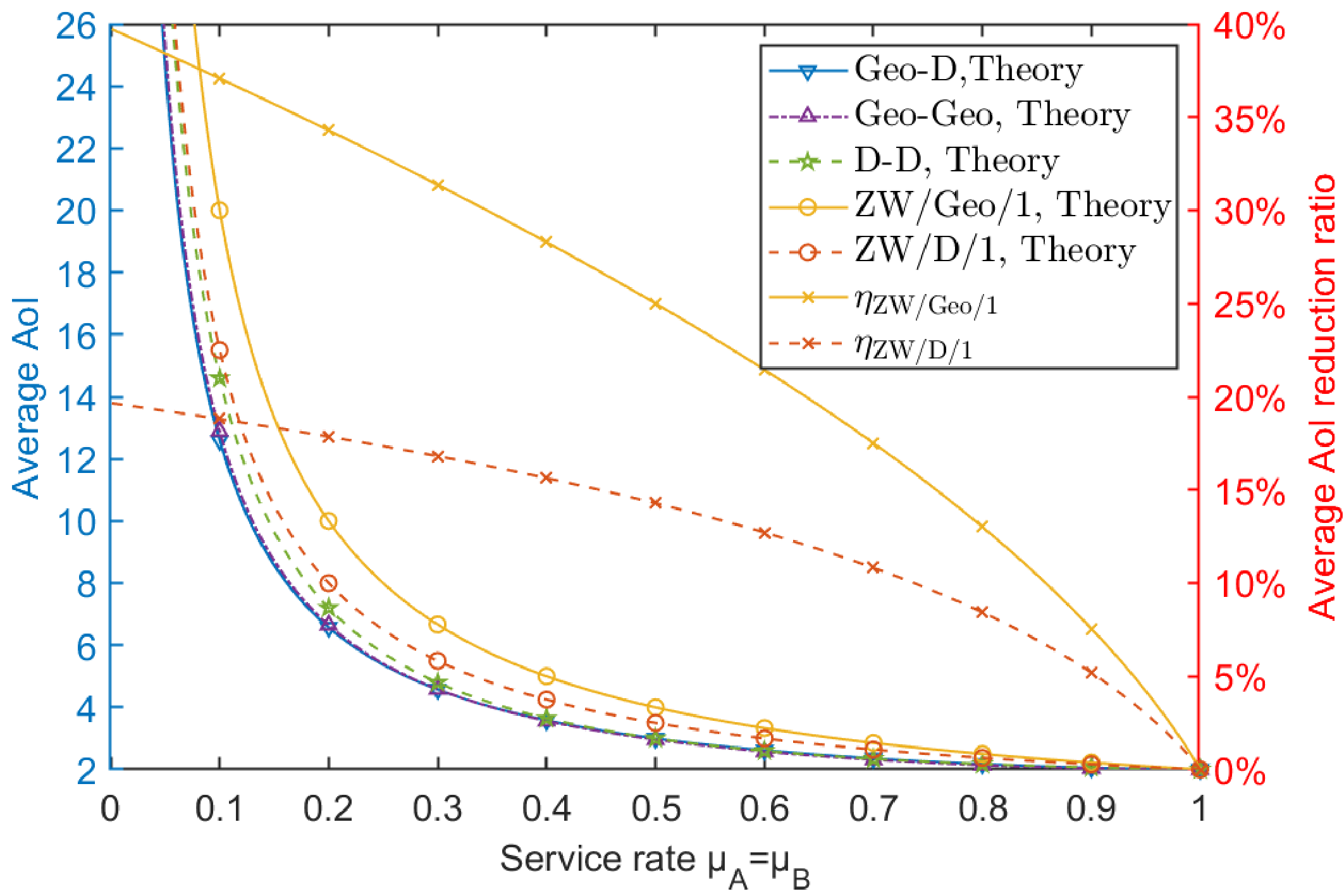}} 
				\caption{Average AoI and reduction ratio of different systems versus service rate $\mu_{\text A}=\mu_{\text {B}}$.}
				\label{aoiu}
			\end{figure}	
			
			To explore the information freshness of the Geo-D system and the Geo-Geo system, we present the average AoI normalized with ${\overline{\Delta}_{\text{A,ZW/D/1}}}=2/\mu_{\text{A}}$ versus the service rate ratio $\mu_{\text{B}}/\mu_{\text{A}}$ in Fig. \ref{dual}. 
			%In order to verify the above analysis of performance difference between Geo-D system and Geo-Geo system, we study the  variation of AoI from another perspective. Fig. \ref{dual}  show the variation of  average  AoI being scaled by $\mu_{\text A}$ with the variation of $\mu_{\text B}/\mu_{\text {A}}$ of  Geo-D and Geo-Geo systems. 
			It is seen that when the service rate is high, i.e., $\mu_{\text{A}}=0.5$, the Geo-Geo system can achieve a small average AoI compared with the Geo-D system. This manifests that
			for a high service rate, a high randomness of both queues can be crucial for achieving  good information freshness. In contrast, when the service rate is low, i.e., $\mu_{\text{A}}=0.1$, the Geo-D system would perform better than the Geo-Geo system when $\mu_{\text {B}}/\mu_{\text A}>0.85$, i.e., when the service rate difference between the two queues is small. This indicates that when the server of both queues are weak but competitive, a periodic status update flow will help  maintain a low average AoI compared with a Bernoulli process. The observed significant effect of the statistics of each queue on the information freshness of the dual-queue system reveals that efforts in carefully designing the component queue for the dual-queue system are extremely valuable.  
			
			%It corresponds to the reason of the variation of curve in Fig. \ref{aoiu}, that ZW/D/1 queue providing stable updates at low service rate and ZW/Ber/1 queue providing more updates at high service rate.
			\begin{figure}[htp]    %AOI的图像
				%			\subfloat[Average AoI of dual-queue systems when  service rate is at high level \label{high}]{\includegraphics[width=0.9\textwidth]{high.eps}}  
				%			\subfloat[Average AoI of dual-queue systems when  service rate is at high level \label{high}]{\includegraphics[width=0.9\textwidth]{matfig4.eps}}  
				\subfloat{\includegraphics[width=0.9\textwidth]{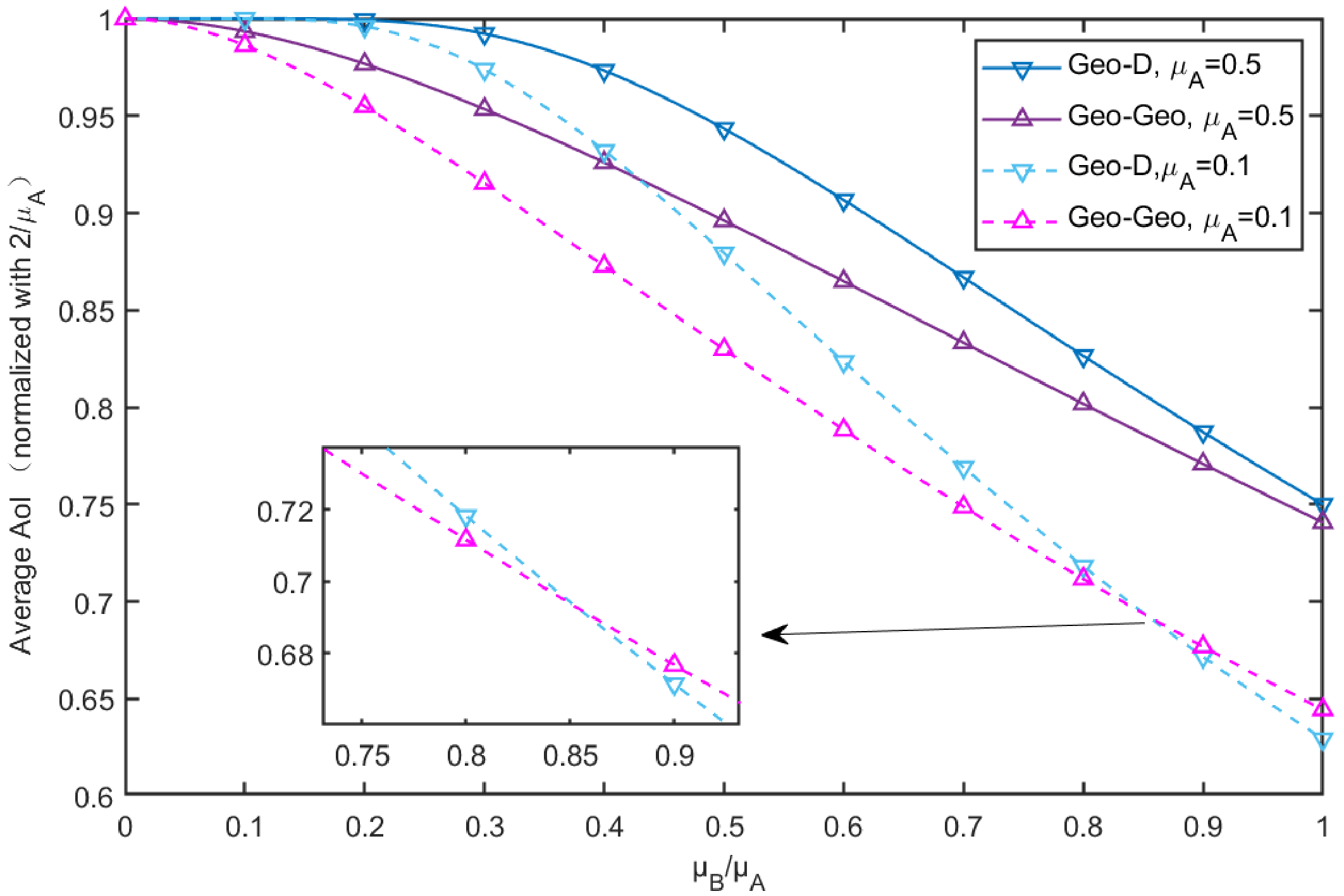}}  
				%			\hfill 	
				%			\subfloat[Average AoI of dual-queue systems when service rate is at low level  \label{low}]{\includegraphics[width=0.9\textwidth]{low.eps}}  
				%			\subfloat[Average AoI of dual-queue systems when service rate is at low level  \label{low}]{\includegraphics[width=0.9\textwidth]{matfig6.eps}}  
				\caption{Normalized average AoI  of dual-queue systems versus service rate ratio $\mu_{\text {B}}/\mu_{\text A}$. }
				\label{dual}
			\end{figure}

			%We compare the performance difference between  single queues and Geo-D systems in the same way, i.e., fixing the service rate of one  sensor to study the performance gains from additional queues.
			%  Fig. \ref{aoiuab} shows the variation of AoI reduction ratio compared with single queue systems with the variation of $\mu_{\text A}/\mu_{\text {B}}$ and $\mu_{\text {B}}/\mu_{\text A}$.
			Fig. \ref{aoiuab} presents the reduction ratio of the average AoI with unequal service rates versus the service rate ratio. The obsolete update ratio, i.e., obsolete updates to all updates, is also illustrated. It is clear from Fig. \ref{aoiuab} that increasing the service rate of an additional queue leads to an increasing average AoI reduction ratio,   as expected. This is achieved by a more frequent arrival of valid status updates  with a penalty of an increasing obsolete update ratio.  It is valuable to see that when the service rate ratio is small (resp. large), the reduction ratio of adding a ZW/Geo/1 queue (resp. a ZW/D/1 queue) would be more significant than adding a ZW/D/1 queue (resp. a ZW/Geo/1 queue). This supports the usefulness of both kinds of queues. Besides, it is also found that if the two service rates    shrink proportionally, the reduction ratio increases. This manifests that the average AoI gain brought by the dual-queue system is relatively significant  for low service rates.  	%When fixing the service rate of a queue, increasing the service rate of the additional queue can enable more status updates to be transmitted to the system, reducing average AoI.  On the other hand, due to more status updates being sent to monitor, which only select fresher status updates to update and stale status updates are obsoleted. As service rates increase, AoI certainly decreases, but leading to an increment on obsolete packet ratio. %Meanwhile, if the service rates of both queues are increased proportionally, it  can be seen that the average AoI reduce but the reduction ratio decreases. 	 This indicates that in low service rate communication scenarios, the AoI performance gain brought by the dual-queue system is relatively significant.
			%	By taking the boundary value of the parameters, it can be observed that the expression of the dual-queue degenerates into the expression of a single queue, no longer bringing any AoI reduction, as shown in the Table \ref{trend}.  
			
			\begin{figure}[htp] 
				\hfill 	
				%			\subfloat{\includegraphics[width=0.9\textwidth]{aoiuab.eps}} 
				\subfloat{\includegraphics[width=0.9\textwidth]{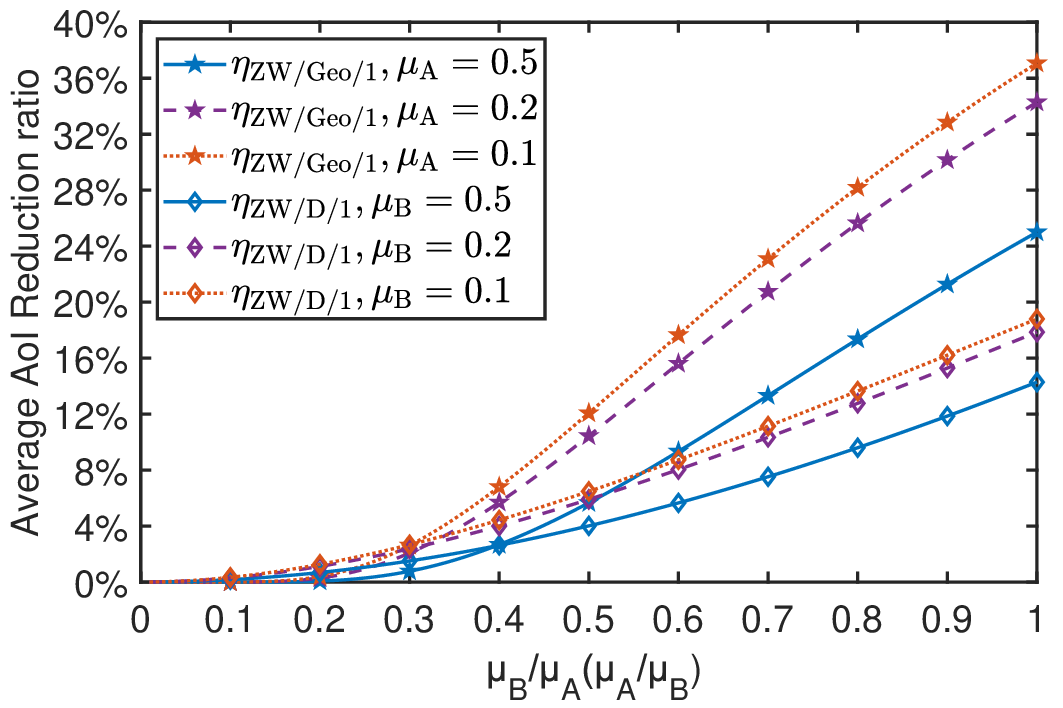}} 
				\caption{Average AoI reduction ratio versus service rate ratio  $\mu_{\text {B}}/\mu_{\text A}$ ($\mu_{\text A}/\mu_{\text {B}}$).}
				\label{aoiuab}
			\end{figure}

			Fig. \ref{paoi} compares the average PAoI performance of the considered Geo-D system with other related queues. Specifically, the average PAoI of the Geo-Geo system and that of the D-D system are obtained from simulations, i.e., taking the mean of 10 rounds of 5000 time slots.
			%			According to the total status updates served  by two sensors and the valid  status updates received by the monitor, the obsolete update ratio  is calculated.
			Comparing  dual-queue systems, it is seen from Fig. \ref{paoi}\subref{paoiu} that the Geo-Geo system, which has the highest level of randomness in the service time, performs the best in terms of the average PAoI, followed by the Geo-D system and the D-D system, where the latter possesses no randomness in service time. This is because the system with a higher level of randomness provides the possibility to distribute the value of the AoI peak, reducing the average PAoI for some cases.
			% Geo-Geo system and the Geo-D system achieve a better PAoI than the D-D system. 
			Besides, similar to the average AoI, the dual-queue systems outperform the single-queue system in terms of the PAoI. This also benefits from the increased valid status updates the parallel updating brings. 
			%				While the PAoI of the two single queue is the same, the reduction ratio of single-queue system also experiences a continuous increasing when the service rate is decreasing. This is also contributed from the long arrival interval due to small service rate and anachronism of the two service processes. 
			Fig. \ref{paoi}\subref{paoiuab} further presents how the reduction ratios change with the service rate ratio in detail. Different from the performance in terms of the average AoI, one can see that adding a ZW/Geo/1 queue for a ZW/D/1 queue would  lead to a reduction of PAoI, while the opposite would not always. Adding a relatively small service rate of ZW/D/1 queue ($\mu_{\text{B}}/\mu_{\text{A}}<0.4$) leads to the at most up to 1.3\% increase of the average PAoI. This can be explained as follows. When $\mu_{\text B}/\mu_{\text{A}}$ is small, the long deterministic service time of sensor B contributes more frequently to the peak AoI of the dual-queue system,  reshaping  the distribution of PAoI and resulting in  a slight increase in the average PAoI. When $\mu_{\text B}/\mu_{\text{A}}$ becomes large, more frequent valid status updates finally lead to the decrease of the average PAoI of the Geo-D system.

			\begin{figure}[htp]    %PAOI的图像
				\subfloat[Average PAoI and reduction ratio versus service rate $\mu_{\text A}=\mu_{\text {B}}$  variation\label{paoiu}]{\includegraphics[width=0.9\textwidth]{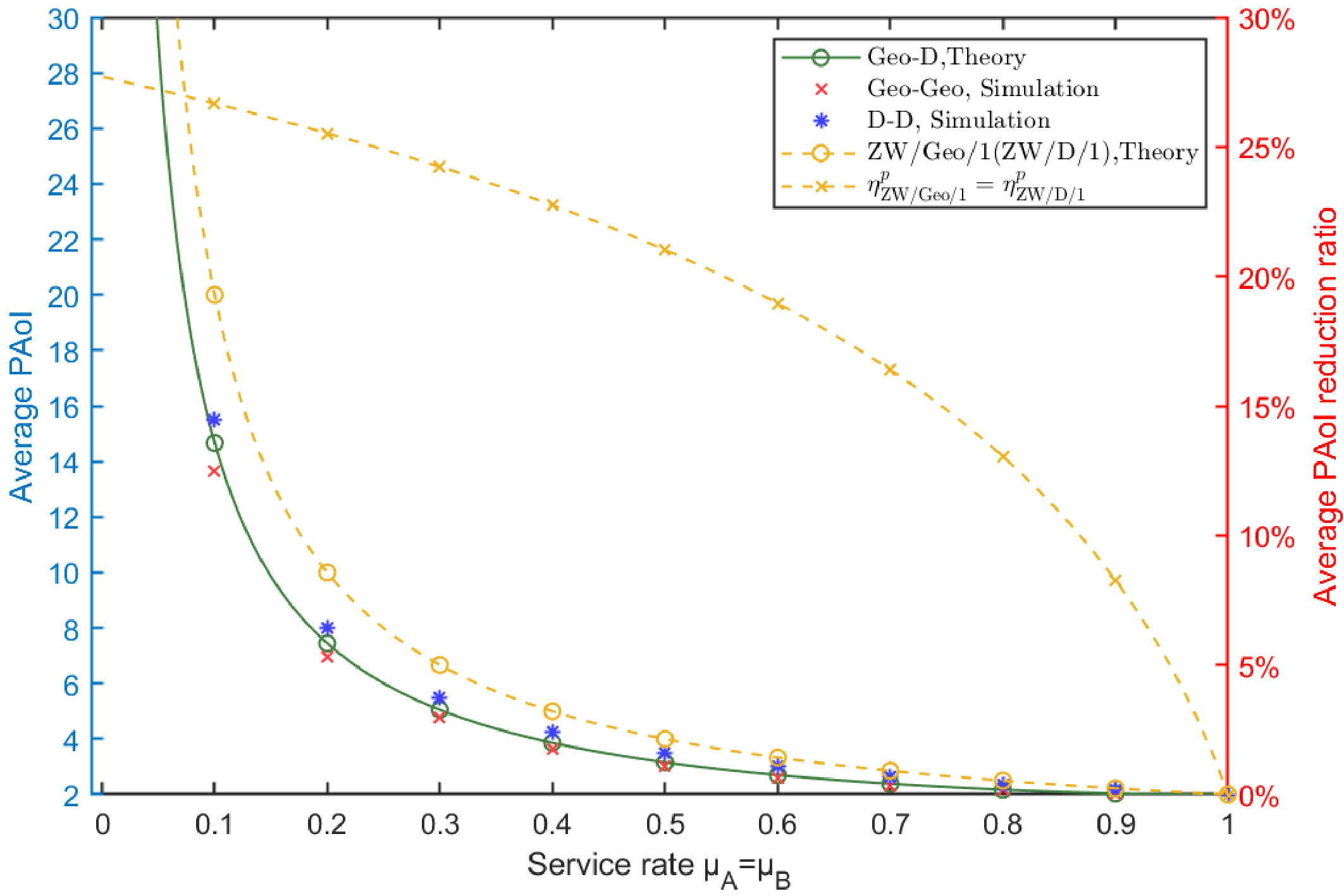}} 
				\hfill 	
				\subfloat[Average PAoI reduction ratio versus service rate ratio $\mu_{\text {B}}/\mu_{\text A}$ \label{paoiuab}]{\includegraphics[width=0.9\textwidth]{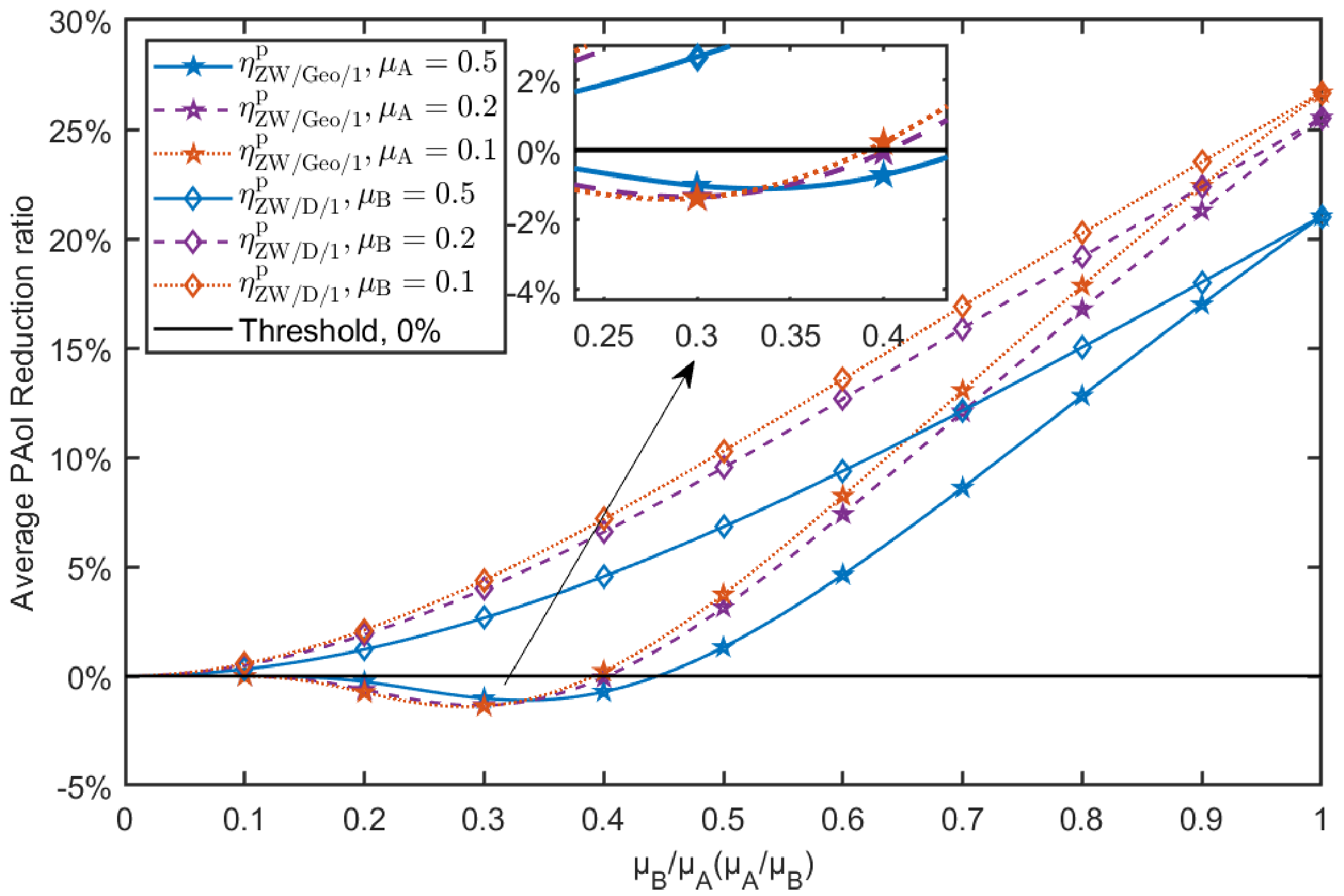}}  
				\caption{Comparison of the Geo-D system with other benchmark queues in terms of average  PAoI. }
				\label{paoi}
			\end{figure}
			
			Finally, Fig. \ref{approximation} compares the average AoI of the discrete-time system and the continuous-time counterpart for different time slot length factor $\delta$. For both the Geo-D system and the Geo-Geo system, the average AoI in second (multiplied with time length $1/\delta$) quite approaches  that of the M-D   and M-M systems when $\delta=1000$. This validates the effectiveness of the results presented in Theorem \ref{T3} and Theorem \ref{T4}, respectively. There are another two key observations  to be highlighted in Fig. \ref{approximation}. First, when $\delta=100$, there is a significant gap between the discrete and  continuous-time systems, especially for large service rates. This implies that dedicated design should be  considered for discrete-time and continuous-time systems to achieve good information freshness. Second, when $\delta$ increases, the convergence from the Geo-Geo system to the M-M system is  faster than that from the Geo-D system to the M-D system. Hence, in cases where one needs to approximate the continuous-time system based on the discrete-time system in terms of average AoI, an appropriate time slot should be adopted to make the approximation  accurate. 
			\begin{figure}[htp]   %逼近
				\subfloat{\includegraphics[width=0.9\textwidth]{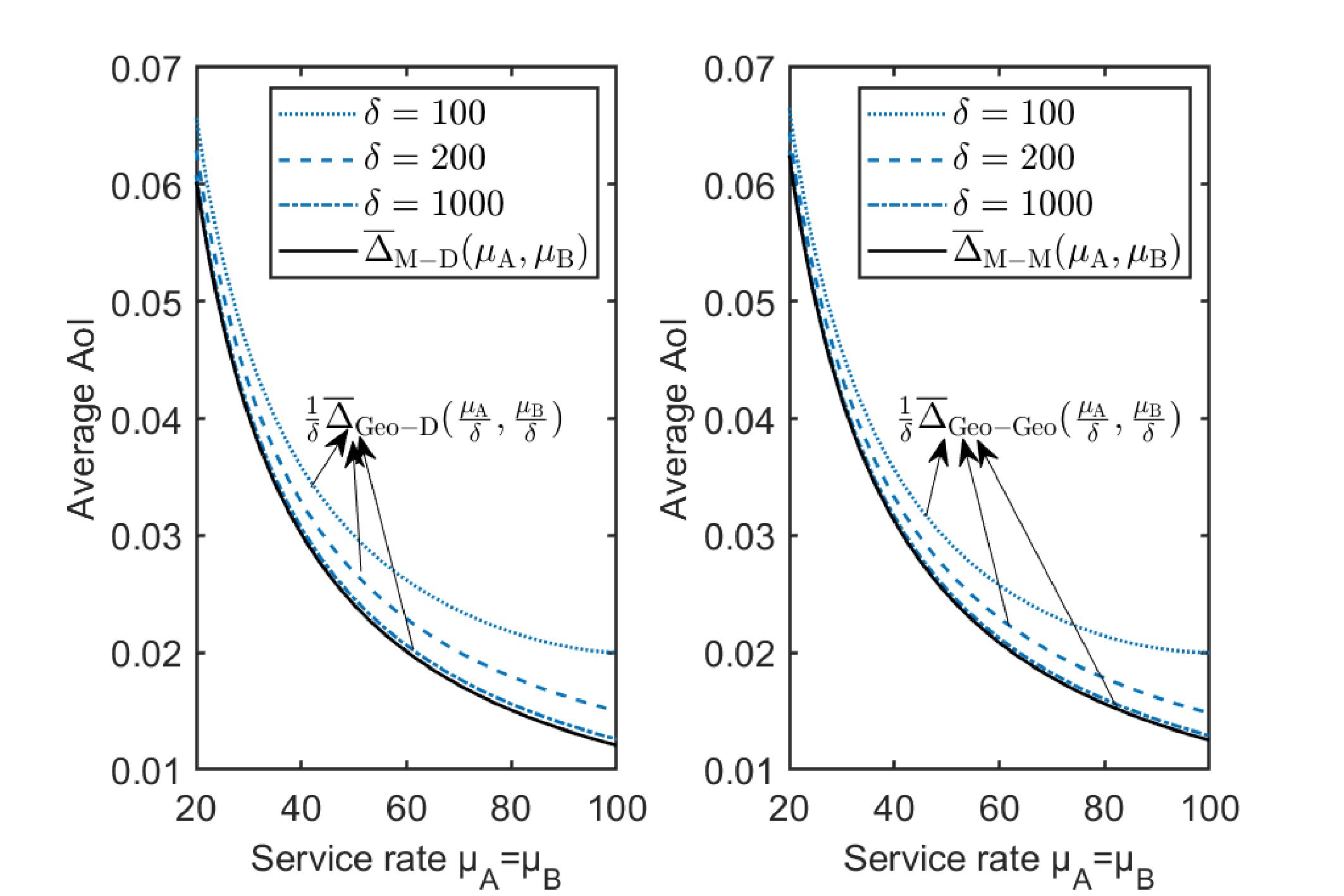}} 
				\caption{The average AoI  of discrete-time system and the counterpart continuous-time system versus service rate $\mu_{\text{A}}=\mu_{\text{B}}$. }
				\label{approximation} 
			\end{figure}			
			\section{ Conclusion}\label{V}
			We investigated the average  AoI and PAoI of the Geo-D system in which the service time of one of the two sensors is geometrically distributed and that of the other is deterministic. Closed-form expressions for the average PAoI and AoI of the Geo-D systems are derived using the graphical analysis method. The average AoI of the Geo-Geo system is  also presented. The connection between the discrete-time system and the counterpart continuous-time system is  studied  based on the obtained results. It is theoretically proven that compared with the ZW/Geo/1 queue and the ZW/D/1 queue,  up to   39.74\% and 19.65\% reduction can be achieved in  average AoI, respectively. Besides,  up to   27.71\% reduction can be achieved in average PAoI for both component single-queue systems. Compared with the Geo-Geo system and the D-D system, the Geo-D system possesses a good balance between randomness and determinacy,  leading to a good performance of information freshness. Finally, it is also proven that the average AoI of the M-M  and the M-D systems are just a limit case of the average AoI of the Geo-Geo  and the Geo-D systems, respectively. Hence, the AoI results obtained for discrete-time systems hold a strong generality and provide a good foundation for status update system design. %the Geo-D queue system effectively reduces the average AoI and PAoI, especially for low service rates. In terms of the dual-queue system, we compare the average AoI performance of the Geo-D queue system with the Geo-Geo queue system in which the service time of both sensors is geometrically distributed and the D-D queue system in which the service time of both sensors is determined. 
			%	The results show that dual-queue systems have their own advantages under different service rates.   Mainly because the ZW/D/1 queue provides stable updates at a low service rate, and the ZW/Ber/1 queue provides more updates at a high service rate.
			% Research shows that dual-queue systems can effectively receive more update status updates, further reducing AoI. However, at the same time, it brings additional network resource overhead and increases the obsolete packet ratio. 
			% On the other hand, based on the relationship between geometric and exponential distributions, we investigate the association 
			%of average AoI between the discrete-time and continuous-time systems. It is found that the average  AoI expressions of continuous-time systems are a limit case of that of discrete-time systems, i.e., the time slot length tends to zero, which shows that the study of discrete-time systems is more universal and detailed.		
			\appendix
				\subsection{ Proof of Lemma 1}\label{A}
			First, let's consider such a formula:
			\begin{equation}
				\sum_{x_1=1}^{T-n+1}\sum_{x_2=1}^{T-n+2-x_1}\sum_{x_3=1}^{T-\sum_{i=1}^{2}x_i-(n-3)}...\sum_{x_m=1}^{T-\sum_{i=1}^{m-1}x_i-(n-m)}...\sum_{x_n=1}^{T-\sum_{i=1}^{n-1}x_i} 1 .
			\end{equation}
			In order to calculate this formula, we first use the sum in reverse order to simplify it.
			\begin{equation}
				\sum_{x=n}^{m}f(x)=\sum_{x=n}^{m}f(m+n-x).
			\end{equation}
			If the sum of each layer is in reverse order, the above formula can be simplified as:
			\begin{equation}
				\sum_{x_1=1}^{T-n+1}\sum_{x_2=1}^{x_1}...\sum_{x_n=1}^{x_{n-1}} 1.
			\end{equation}
			By simple layer by layer recursion, the following formula can be obtained. The recursive process is omitted here.
			\begin{equation}
				\label{mn}
				\sum_{x_1=1}^{m}\sum_{x_2=1}^{x_1}...\sum_{x_n=1}^{x_{n-1}} 1=\tbinom{m+n-1}{n}.
			\end{equation}
			In the formula, $m$ is the upper limit of the outermost variable, and $n$ is the number of layers of the summation structure.
			Now prove lemma 1.
			\begin{equation}
				\sum_{x_1=1}^{T-n+1}...\sum_{x_i=1}^{T-(n-i)-\sum_{m=1}^{i-1}x_m}...\sum_{x_n=1}^{T-\sum_{i=1}^{n-1}x_i}x_i=\tbinom{T+1}{n+1}.
			\end{equation}
			Using the inverse summation and \eqref{mn}, the original equation is reduced to
			\begin{align}
				&\sum_{x_1=1}^{T-n+1}...\sum_{x_i=1}^{T-(n-i)-\sum_{m=1}^{i-1}x_m}x_i\tbinom{T-\sum_{m=1}^{i-1}}{n-i}=\sum_{x_1=1}^{T-n+1}...\sum_{x_{i-1}=1}^{T-(n-i+1)-\sum_{m=1}^{i-2}x_m}  	\Big\{ \tbinom{T-\sum_{m=1}^{i-1}x_m+1}{n-i+2}                      \nonumber \\
				&+2\tbinom{T-\sum_{m=1}^{i-1}x_m+1}{n-i+2}       ...         +(T-\sum_{m=1}^{i-1}x_m-n+i) \tbinom{n-i}{n-i}      \Big\}  \nonumber \\
				&=\sum_{x_1=1}^{T-n+1}...\sum_{x_{i-1}=1}^{T-(n-i+1)-\sum_{m=1}^{i-2}x_m}\tbinom{T-\sum_{m=1}^{i-1}x_m+1}{n-i+2}             \nonumber \\
				&=\sum_{x_1=1}^{T-n+1}...\sum_{x_{i-2}=1}^{T-(n-i+2)-\sum_{m=1}^{i-3}x_m} \tbinom{T-\sum_{m=1}^{i-2}x_m+2}{n-i+3}    \nonumber \\
				&=\tbinom{T+1}{n+1}.
			\end{align}
			From the formula, it can be seen that the order does not affect the equivalence.
			\begin{equation}
				\sum_{x_1=1}^{T-n+1}...\sum_{x_i=1}^{T-(n-i)-\sum_{m=1}^{i-1}x_m}...\sum_{x_n=1}^{T-\sum_{i=1}^{n-1}x_i}x_i=\tbinom{T+1}{n+1},
			\end{equation}where $i=1,2,..,n$. This completes the proof of Lemma \ref{lemma1}.
			\subsection{ Proof of Lemma 2}\label{B}
			To prove the following lemma:
			\begin{equation}
				\sum_{x_1=1}^{T-n+1}...\sum_{x_i=1}^{T-(n-i)-\sum_{m=1}^{i-1}x_m}...\sum_{x_n=1}^{T-\sum_{i=1}^{n-1}x_i}x_i^{2}=\tbinom{T+1}{n+1}+2\tbinom{T+1}{n+2},
			\end{equation}taking $x_1$ as an example.By adopting the same proof method in Lemma 1, we can get:
			\begin{equation}
				\sum_{x_1=1}^{T-n+1}...\sum_{x_i=1}^{T-(n-i)-\sum_{m=1}^{i-1}x_m}...\sum_{x_n=1}^{T-\sum_{i=1}^{n-1}x_i}x_1^{2}=\sum_{x_1=1}^{T-n+1}\tbinom{T-x_1}{n-1}x_1^{2}.
			\end{equation}
			The formula can be expanded as follows:
			\begin{align}
				\sum_{x_1=1}^{T-n+1}\tbinom{T-x_1}{n-1} x_1^{2}&=\tbinom{T-1}{n-1}     +2^2  \tbinom{T-2}{n-1}    +...+(T-n+1)^2 \tbinom{n-1}{n-1}\nonumber \\
				&=\tbinom{T}{n}+(2^2-1)\tbinom{T-2}{n-1}+...+((T-n+1)^2-1)\tbinom{n-1}{n-1}   \nonumber  \\
				&=\tbinom{T}{n} +3 \tbinom{T-1}{n}+5 \tbinom{T-2}{n}+... +(2T-2n+1) \tbinom{n}{n}        \nonumber \\
				&=\tbinom{T+1}{n+1}+2\tbinom{T+1}{n+2}.
			\end{align}
			Similarly, the order does not affect the equivalence, and the following formula is proved:
			\begin{equation}
				\sum_{x_1=1}^{T-n+1}...\sum_{x_i=1}^{T-(n-i)-\sum_{m=1}^{i-1}x_m}...\sum_{x_n=1}^{T-\sum_{i=1}^{n-1}x_i}x_i^{2}=\tbinom{T+1}{n+1}+2\tbinom{T+1}{n+2},
			\end{equation}Where $i=1,2,..,n.$ This completes the proof of Lemma \ref{lemma2}.
			\subsection{ Proof of Lemma 3}\label{C}
			To prove the following lemma:
			\begin{equation}
				\sum_{x_1=1}^{T-n+1}...\sum_{x_i=1}^{T-(n-i)-\sum_{m=1}^{i-1}x_m}...\sum_{x_n=1}^{T-\sum_{i=1}^{n-1}x_i}x_ix_j=\tbinom{T+2}{n+2},
			\end{equation}
			taking $x_1$ and $ x_2 $ as an example. By adopting the same proof method in Lemma 1, we can get:
			\begin{equation}
				\sum_{x_1=1}^{T-n+1}...\sum_{x_i=1}^{T-(n-i)-\sum_{m=1}^{i-1}x_m}...\sum_{x_n=1}^{T-\sum_{i=1}^{n-1}x_i}x_1x_2=\sum_{x_1=1}^{T-n+1}\sum_{x_2=1}^{T-x_1-n+2}\tbinom{T-x_1-x_2}{n-2}x_1x_2.
			\end{equation}
			The formula can be expanded as follows:
			\begin{align}
				\sum_{x_1=1}^{T-n+1}\sum_{x_2=1}^{T-x_1-n+2}\tbinom{T-x_1-x_2}{n-2}x_1x_2&=\sum_{x_1=1}^{T-n+1}x_1\sum_{x_2=1}^{T-x_1-n+2}x_2\tbinom{T-x_1-x_2}{n-2} \nonumber \\
				%&=\sum_{x_1=1}^{T-n+1}x_1[C_{T-x_1-1}^{n-2}+2C_{T-x_1-2}^{n-2}...+(T-n+2-x_1)C_{n-2}^{n-2}                                    ]\nonumber \\
				&=\sum_{x_1=1}^{T-n+1}x_1\tbinom{T-x_1+1}{n}=\tbinom{T+2}{n+2}.
			\end{align}
			This completes the proof of Lemma \ref{lemma3}.

			\baselineskip=18pt
			%			\bibliographystyle{IEEEtran}
			%			\bibliography{IEEEtran,test}
			%			\bibliographystyle{IEEEtran}
			
%			
%			\bibliographystyle{IEEEtran}
%			\bibliography{IEEEabrv,test}
			
			% Generated by IEEEtran.bst, version: 1.14 (2015/08/26)

			%			\bibliography{test}

\begin{thebibliography}{10}
				\providecommand{\url}[1]{#1}
				\csname url@samestyle\endcsname
				\providecommand{\newblock}{\relax}
				\providecommand{\bibinfo}[2]{#2}
				\providecommand{\BIBentrySTDinterwordspacing}{\spaceskip=0pt\relax}
				\providecommand{\BIBentryALTinterwordstretchfactor}{4}
				\providecommand{\BIBentryALTinterwordspacing}{\spaceskip=\fontdimen2\font plus
					\BIBentryALTinterwordstretchfactor\fontdimen3\font minus
					\fontdimen4\font\relax}
				\providecommand{\BIBforeignlanguage}[2]{{%
						\expandafter\ifx\csname l@#1\endcsname\relax
						\typeout{** WARNING: IEEEtran.bst: No hyphenation pattern has been}%
						\typeout{** loaded for the language `#1'. Using the pattern for}%
						\typeout{** the default language instead.}%
						\else
						\language=\csname l@#1\endcsname
						\fi
						#2}}
				\providecommand{\BIBdecl}{\relax}
				\BIBdecl
				
				\bibitem{Iot}
				J.~Lin, W.~Yu, N.~Zhang, X.~Yang, H.~Zhang, and W.~Zhao, ``{A} survey on
				{Internet} of things: {Architecture}, enabling technologies, security and
				privacy, and applications,'' \emph{IEEE Internet Things J.}, vol.~4, no.~5,
				pp. 1125--1142, Oct 2017.
				
				\bibitem{Aoigainian}
				S.~Kaul, M.~Gruteser, V.~Rai, and J.~Kenney, ``{Minimizing age of information
					in vehicular networks},'' in \emph{Proc. 8th Annu. IEEE Commun. Soc. Conf.
					Sensor, Mesh Ad-Hoc Commun. Netw. (SECON)}, June 2011, pp. 350--358.
				
				\bibitem{firstwork}
				S.~Kaul, R.~Yates, and M.~Gruteser, ``{Real-time status: How often should one
					update?}'' in \emph{Proc. IEEE Int. Conf. on Computer. Commun. (INFOCOM)},
				2012, pp. 2731--2735.
				
				\bibitem{stationary}
				Y.~Inoue, H.~Masuyama, T.~Takine, and T.~Tanaka, ``{A} general formula for the
				stationary distribution of the age of information and its application to
				single-server queues,'' \emph{IEEE Trans. Inf. Theory}, vol.~65, no.~12, pp.
				8305--8324, Dec 2019.
				
				\bibitem{distribution1}
				L.~Huang and E.~Modiano, ``{Optimizing age-of-information in a multi-class
					queueing system},'' in \emph{2015 IEEE Int. Symp. on Inf. Theory (ISIT)},
				June 2015, pp. 1681--1685.
				
				\bibitem{distribution2}
				R.~D. Yates, ``{The} age of information in networks: Moments, distributions,
				and sampling,'' \emph{IEEE Trans. Inf. Theory}, vol.~66, no.~9, pp.
				5712--5728, Sep. 2020.
				
				\bibitem{distribution3}
				R.~Talak, S.~Karaman, and E.~Modiano, ``{Optimizing} information freshness in
				wireless networks under general interference constraints,'' \emph{IEEE/ACM
					Trans. Netw.}, vol.~28, no.~1, pp. 15--28, Feb 2020.
				
				\bibitem{LCFS}
				S.~K. Kaul, R.~D. Yates, and M.~Gruteser, ``Status updates through queues,'' in
				\emph{2012 46th Annu. Conf. on Inf. Sci. and Syst. (CISS)}, March 2012, pp.
				1--6.
				
				\bibitem{bao1}
				M.~Costa, M.~Codreanu, and A.~Ephremides, ``{On} the age of information in
				status update systems with packet management,'' \emph{IEEE Trans. Inf.
					Theory}, vol.~62, no.~4, pp. 1897--1910, April 2016.
				
				\bibitem{bao2}
				Z.~Tang, N.~Yang, X.~Zhou, and J.~Lee, ``{Average} age of information penalty
				of short-packet communications with packet management,'' in \emph{2023 IEEE
					Int. Conf. Commun.(ICC)}, May 2023, pp. 1670--1675.
				
				\bibitem{bao3}
				P.~Zou, O.~Ozel, and S.~Subramaniam, ``{Waiting} before serving: {A} companion
				to packet management in status update systems,'' \emph{IEEE Trans. Inf.
					Theory}, vol.~66, no.~6, pp. 3864--3877, June 2020.
				
				\bibitem{firstMulti}
				R.~D. Yates and S.~Kaul, ``{Real-time} status updating: {Multiple} sources,''
				in \emph{2012 IEEE Int. Symp. Inf. Theory}, July 2012, pp. 2666--2670.
				
				\bibitem{duoduilie1}
				C.~Kam, S.~Kompella, G.~D. Nguyen, and A.~Ephremides, ``{Effect} of message
				transmission path diversity on status age,'' \emph{IEEE Trans. Inf. Theory},
				vol.~62, no.~3, pp. 1360--1374, March 2016.
				
				\bibitem{duoduilie2}
				R.~D. Yates, ``{Status} updates through networks of parallel servers,'' in
				\emph{2018 IEEE Int. Symp. Inf. Theory (ISIT)}, June 2018, pp. 2281--2285.
				
				\bibitem{TCP}
				K.~Gao, C.~Xu, X.~Ji, J.~Qin, S.~Yang, L.~Zhong, and D.~Wu, ``Freshness-aware
				age optimization for multipath {TCP} over software defined networks,''
				\emph{IEEE Trans. Network Sci. Eng.}, 2021, doi:{10.1109/TNSE.2021.3075704}.
				
				\bibitem{shuangduilie}
				X.~Jia, S.~Cao, and M.~Xie, ``{Age} of information of dual-sensor information
				update system with {HARQ} chase combining and energy harvesting diversity,''
				\emph{IEEE Wireless Commun. Lett.}, vol.~10, no.~9, pp. 2027--2031, Sep.
				2021.
				
				\bibitem{mmmd}
				D.~Deng, Z.~Chen, H.~H. Yang, N.~Pappas, L.~Hu, M.~Wang, Y.~Jia, and T.~Q.
				Quek, ``Information freshness in a dual monitoring system,'' in \emph{2022
					IEEE Global Commun. Conf.(GlobeCom)}, Dec 2022, pp. 4977--4982.
				
				\bibitem{wuren1}
				X.~Wang, M.~Yi, J.~Liu, Y.~Zhang, M.~Wang, and B.~Bai, ``Cooperative data
				collection with multiple {UAVs} for information freshness in the internet of
				things,'' \emph{IEEE Trans. Commun.}, vol.~71, no.~5, pp. 2740--2755, May
				2023.
				
				\bibitem{wuren2}
				W.~Jiang, B.~Ai, J.~Cheng, Y.~Lin, and G.~Zhang, ``Sum of age-of-information
				minimization in aerial {IRSs} assisted wireless networks,'' \emph{IEEE
					Commun. Lett.}, vol.~27, no.~5, pp. 1377--1381, May 2023.
				
				\bibitem{aloha_multi}
				Q.~Wang and H.~Chen, ``Age of information in reservation multi-access networks
				with stochastic arrivals: {Analysis} and optimization,'' \emph{IEEE Trans.
					Commun.}, vol.~71, no.~8, pp. 4707--4720, Aug. 2023.
				
				\bibitem{aloha_col}
				Y.~Lai, T.-T. Chan, J.~Liang, and H.~Pan, ``Age of information in multichannel
				slotted {ALOHA}: {Should} collided users send first?'' in \emph{2023 IEEE
					13th Annu. Comput. Commun. Workshop and Conf. (CCWC)}, March 2023, pp.
				1212--1217.
				
				\bibitem{aloha_chan}
				X.~Chen, K.~Gatsis, H.~Hassani, and S.~S. Bidokhti, ``Age of information in
				random access channels,'' \emph{IEEE Trans. Inf. Theory}, vol.~68, no.~10,
				pp. 6548--6568, Oct 2022.
				
				\bibitem{aloha_pa}
				S.~Asvadi and F.~Ashtiani, ``Peak age of information in slotted {ALOHA}
				networks,'' \emph{IEEE Trans. Commun.}, vol.~71, no.~10, pp. 6018--6030, Oct
				2023.
				
				\bibitem{lisanfcfs}
				R.~Talak, S.~Karaman, and E.~Modiano, ``Optimizing information freshness in
				wireless networks under general interference constraints,'' \emph{IEEE/ACM
					Trans. Netw.}, vol.~28, no.~1, pp. 15--28, Feb 2020.
				
				\bibitem{lisanber}
				J.~Zhang and Y.~Xu, ``On age of information for discrete time status updating
				system with {Ber/G/1/1} queues,'' in \emph{2021 IEEE Inf. Theory Workshop
					(ITW)}, Oct 2021, pp. 1--6.
				
				\bibitem{lisanlcfs}
				A.~Kosta, N.~Pappas, A.~Ephremides, and V.~Angelakis, ``Queue management for
				age sensitive status updates,'' in \emph{2019 IEEE Int. Symp. on Inf. Theory
					(ISIT)}, July 2019, pp. 330--334.
				
				\bibitem{discretestationary}
				------, ``The age of information in a discrete time queue: {Stationary}
				distribution and non-linear age mean analysis,'' \emph{IEEE J. Sel. Areas
					Commun.}, vol.~39, no.~5, pp. 1352--1364, May 2021.
				
				\bibitem{discreteMultiple}
				N.~Akar and O.~Dogan, ``{Discrete}-time queueing model of age of information
				with multiple information sources,'' \emph{IEEE Internet Things J.}, vol.~8,
				no.~19, pp. 14\,531--14\,542, Oct 2021.
				
				\bibitem{zerowait}
				Y.~Sun, E.~Uysal-Biyikoglu, R.~D. Yates, C.~E. Koksal, and N.~B. Shroff,
				``{Update} or wait: {How} to keep your data fresh,'' \emph{IEEE Trans. Inf.
					Theory}, vol.~63, no.~11, pp. 7492--7508, Nov 2017.
				
				\bibitem{geogeo}
				T.~Zhang, Z.~Chen, Z.~Tian, L.~Zhen, Y.~Jia, M.~Wang, D.~O. Wu, and T.~Q.~S.
				Quek, ``On the timeliness of the stalest stream among multiple status
				updating streams,'' \emph{IEEE Internet Things J.}, pp. 13\,622--13\,635,
				2023.
				
			\end{thebibliography}
%			\newpage
%			\vfill
%			
%			
%			
%			\ifCLASSOPTIONcaptionsoff
%			\newpage
%			\fi
		\end{document}